\documentclass[11pt]{article}
\pdfoutput=1

\usepackage[margin=1in]{geometry}
\usepackage[utf8]{inputenc}
\usepackage{amsmath, amsthm, amsfonts}
\usepackage{cite}
\usepackage{hyperref}
\usepackage{authblk}
\usepackage[ruled, vlined, linesnumbered]{algorithm2e}
\usepackage[labelformat=simple]{subcaption}

\usepackage{booktabs}
\usepackage{microtype}
\usepackage{tikz}
\usetikzlibrary{quantikz}

\usepackage{xcolor}

\newtheorem{definition}{Definition}
\newtheorem{theorem}{Theorem}

\newtheorem{problem}{Problem}

\usepackage{tikzit}
\input{zx.tikzdefs}

\tikzstyle{Z dot}=[inner sep=0mm, minimum size=2.25mm, shape=circle, draw=black, fill={rgb,255: red,216; green,248; blue,216}, tikzit fill={rgb,255: red,216; green,248; blue,216}]
\tikzstyle{Z phase dot}=[minimum size=5mm, font={\footnotesize\boldmath}, shape=rectangle, inner sep=0.2mm, outer sep=-2mm, scale=0.8, tikzit shape=circle, draw=black, fill={rgb,255: red,216; green,248; blue,216}, tikzit fill={rgb,255: red,216; green,248; blue,216}, tikzit draw=teal]
\tikzstyle{X dot}=[inner sep=0mm, minimum size=2.25mm, shape=circle, draw=black, fill={rgb,255: red,232; green,165; blue,165}, tikzit fill={rgb,255: red,232; green,165; blue,165}]
\tikzstyle{X phase dot}=[minimum size=5mm, font={\footnotesize\boldmath}, shape=rectangle, inner sep=0.2mm, outer sep=-2mm, scale=0.8, tikzit shape=circle, draw=black, fill={rgb,255: red,232; green,165; blue,165}, tikzit fill={rgb,255: red,232; green,165; blue,165}, tikzit draw=teal]
\tikzstyle{Control node}=[inner sep=0mm, minimum size=1.15mm, shape=circle, draw=black, fill=black]
\tikzstyle{Target node}=[cross, scale=0.8, draw=black, thick, circle]
\tikzstyle{Gate}=[shape=rectangle, thick, minimum width=5.5mm, inner ysep=1.3mm, fill=white, draw=black, font={\small}]
\tikzstyle{Gray node}=[inner sep=0mm, minimum size=2.5mm, font={\tiny}, text=white, fill={rgb,255: red,55; green,65; blue,81}, draw={rgb,255: red,55; green,65; blue,81}, shape=circle]
\tikzstyle{White node}=[inner sep=0mm, minimum size=2.5mm, font={\tiny}, fill=white, draw=black, shape=circle]
\tikzstyle{Graph node}=[inner sep=0mm, minimum size=5mm, font={\footnotesize}, fill=white, draw=black, shape=circle]
\tikzstyle{Label node}=[fill=none, draw=none, shape=circle, inner sep=0mm, font={\scriptsize}, text={rgb,255: red,55; green,65; blue,81}, tikzit draw=blue, tikzit fill=white, tikzit shape=circle]
\tikzstyle{Text node}=[fill=none, text width=2.5cm, align=flush center, draw=none, inner sep=0mm, shape=rectangle, font={\sffamily\scriptsize}, text={rgb,255: red,55; green,65; blue,81}, tikzit draw=blue, tikzit fill=white, tikzit shape=rectangle]
\tikzstyle{Label black}=[fill=none, draw=none, shape=circle, font={\scriptsize}, text=black, tikzit draw=black, tikzit fill=white]
\tikzstyle{Circuit label}=[fill=none, draw=none, shape=circle, font={\small}, text=black, tikzit draw=black, tikzit fill=white]

\tikzstyle{Circuit edge}=[-, thick]
\tikzstyle{Hadamard edge}=[-, dashed, dash pattern=on 2pt off 1pt, thick, draw={rgb,255: red,68; green,136; blue,255}]
\tikzstyle{Dashed grey}=[-, dashed, dash pattern=on 1pt off 3pt, line cap=round, thick, draw={rgb,255: red,55; green,65; blue,81}]
\tikzstyle{Double edge}=[-, double, thick, shorten <=-1mm, shorten >=-1mm, double distance=2pt]
\tikzstyle{Arrow}=[-stealth, very thick, line cap=round, draw={rgb,255: red,55; green,65; blue,81}]
\tikzstyle{Arrow2}=[-stealth, thick, line cap=round, draw=black]
\tikzstyle{Interval}=[|-|, very thick, draw={rgb,255: red,55; green,65; blue,81}]

\title{Qubit-count optimization using ZX-calculus}

\author[1,2]{Vivien Vandaele}
\affil[1]{Eviden Quantum Lab, Les Clayes-sous-Bois, France}
\affil[2]{Université de Lorraine, CNRS, Inria, LORIA, F-54000 Nancy, France}

\date{}

\begin{document}
\maketitle

\begin{abstract}
    We propose several methods for optimizing the number of qubits in a quantum circuit while preserving the number of non-Clifford gates.
    One of our approaches consists in reversing, as much as possible, the gadgetization of Hadamard gates, which is a procedure used by some $T$-count optimizers to circumvent Hadamard gates at the expense of additional qubits.
    We prove the NP-hardness of this problem and we present an algorithm for solving it.
    We also propose a more general approach to optimize the number of qubits by showing how it relates to the problem of finding a minimal-width path-decomposition of the graph associated with a given ZX-diagram.
    This approach can be used to optimize the number of qubits for any computational model that can natively be depicted in ZX-calculus, such as the Pauli Fusion computational model which can represent lattice surgery operations.
    We also show how this method can be used to efficiently optimize the number of qubits in a quantum circuit by using the ZX-calculus as an intermediate representation.
\end{abstract}

\section{Introduction}

Besides their fundamental role in quantum computing, qubits are a crucial resource as they can be used as a trade-off for some other resources.
The most prominent example is provided by the field of quantum error correction: multiple error-prone qubits can be used to form one reliable logical qubit.
Moreover, additional qubits can be valuable for the design and implementation of quantum algorithms as they can sometimes be used to lower the execution time, notably by reducing the computational depth of certain quantum algorithms using various optimization procedures or by enabling their parallelization.
Finally, qubits can also be used as a trade-off for quantum gates.
For instance, implementing the $CCZ$ gate requires $7$ $T$ gates, whereas it can be done with only $4$ $T$ gates by incorporating an ancillary qubit~\cite{jones2013low}.
Also, several algorithms designed for optimizing the number of $T$ gates are heavily relying on ancillary qubits~\cite{heyfron2018efficient, de2020fast, ruiz2024quantum, vandaele2024lower}.

As such, optimizing the number of qubits can enable the execution of a given quantum algorithm if the number of qubits at disposal is insufficient, or it can free up some qubits which can then be used as a trade-off for other resources.
In either case, the optimization of the number of qubits is more beneficial if it is not done at the expense of some other critical resources.
For instance, non-Clifford gates, which are fundamental for universal quantum computation, are critical resrouces since such gates are costly to implement fault-tolerantly in most quantum error-correcting codes.
In this work, we present several methods for efficiently optimizing the number of qubits in quantum circuits and in ZX-diagrams depicting lattice surgery operations, without increasing the number of non-Clifford gates.
Consequently, the circuits generated by our methods have both an optimized number of non-Clifford gates and an optimized number of qubits, two of the most essential resources for fault-tolerant quantum computing.

We present two novel approaches for optimizing the number of qubits in a given quantum circuit.
The first one is based on the degadgetization of Hadamard gates, which is the reverse process of Hadamard gates gadgetization as depicted in Figure~\ref{fig:gadgetization}.
This approach is particularly useful for Clifford$+T$ circuits in which the number of $T$ gates has been optimized.
Indeed, the $T$-count optimizers yielding the best performances are relying on this measurement-based gadget to circumvent Hadamard gates and better optimize the number of $T$ gates~\cite{heyfron2018efficient, de2020fast, ruiz2024quantum, vandaele2024lower}.
In some cases, this can drastically increase the number of qubits required to implement the optimized quantum circuit.
This can be impractical simply because the number of ancillary qubits at disposal may be limited, or it can counteract the gains achieved through $T$-count optimization.
For instance, in the surface code, implementing the $T$ gate fault-tolerantly can be done via magic states distillation, which is a procedure inducing an overhead in time and in the number of qubits~\cite{bravyi2005universal}.
However, it has been shown that for some compilers designed for the surface code, the overhead of magic states distillation can be significantly lower than the overhead induced by the surface code for encoding reliable logical qubits during the whole computation, especially for large quantum circuits~\cite{litinski2019magic}.
In this context, increasing the number of qubits to lower the number of $T$ gates can be less impactful, or, in some cases, counterproductive.

\begin{figure}[t]
    \centering
    \begin{subfigure}{0.59\textwidth} \centering
        \scalebox{0.9}{
            \begin{tikzpicture}
	\begin{pgfonlayer}{nodelayer}
		\node [style=none] (0) at (-0.25, 0) {};
		\node [style=none] (1) at (2.75, 0) {};
		\node [style=Circuit label] (2) at (-1, 0) {$\ket{\psi}$};
		\node [style=Circuit label] (3) at (3.75, 0) {$H\ket{\psi}$};
		\node [style=Gate] (4) at (1.25, 0) {$H$};
		\node [style=none] (5) at (5.5, 0) {$=$};
		\node [style=Control node] (6) at (9.25, 1) {};
		\node [style=Gate] (7) at (9.25, -1) {$Z$};
		\node [style=none] (8) at (7.75, 1) {};
		\node [style=none] (10) at (7.75, -1) {};
		\node [style=none] (11) at (12.75, -1) {};
		\node [style=Gate] (12) at (11.25, -1) {$X$};
		\node [style=Circuit label] (13) at (7, 1) {$\ket{\psi}$};
		\node [style=Circuit label] (14) at (7, -1) {$\ket{+}$};
		\node [style=Circuit label] (15) at (13.75, -1) {$H\ket{\psi}$};
		\node [meter, scale=0.9, yshift=0.16mm] (16) at (11.25, 0.975) {};
		\node [style=none] (17) at (11, 0.775) {\tiny $X$};
	\end{pgfonlayer}
	\begin{pgfonlayer}{edgelayer}
		\draw [style=Circuit edge] (1.center) to (0.center);
		\draw [style=Circuit edge] (6) to (7);
		\draw [style=Circuit edge] (8.center) to (6);
		\draw [style=Circuit edge] (10.center) to (7);
		\draw [style=Circuit edge, in=360, out=180] (11.center) to (7);
		\draw [style=Circuit edge] (16) to (6);
		\draw [style=Double edge] (16) to (12);
	\end{pgfonlayer}
\end{tikzpicture}
        }
        \caption{In the quantum circuit model.\newline}
    \end{subfigure}
    \begin{subfigure}{0.4\textwidth}
        \centering
        \begin{tikzpicture}
	\begin{pgfonlayer}{nodelayer}
		\node [style=Z dot] (0) at (-3.5, 0) {};
		\node [style=Z dot] (1) at (-2, 0) {};
		\node [style=none] (11) at (-4.5, 0) {};
		\node [style=none] (12) at (-1, 0) {};
		\node [style=none] (20) at (0, 0) {$=$};
		\node [style=Z dot] (31) at (2.25, 0.75) {};
		\node [style=Z dot] (32) at (2.25, -0.75) {};
		\node [style=Z phase dot] (33) at (3.5, 0.75) {$s \pi$};
		\node [style=Z dot] (34) at (1.25, -0.75) {};
		\node [style=none] (35) at (1, 0.75) {};
		\node [style=none] (36) at (4.5, -0.75) {};
		\node [style=X phase dot] (37) at (3.5, -0.75) {$s \pi$};
		\node [style=none] (39) at (1, -1.925) {};
	\end{pgfonlayer}
	\begin{pgfonlayer}{edgelayer}
		\draw [style=Hadamard edge] (0) to (1);
		\draw (11.center) to (0);
		\draw (1) to (12.center);
		\draw [style=Hadamard edge] (31) to (32);
		\draw (33) to (31);
		\draw (34) to (32);
		\draw (35.center) to (31);
		\draw (36.center) to (32);
	\end{pgfonlayer}
\end{tikzpicture}
        \caption{In ZX-calculus, where $s \in \{0, 1\}$ is the measurement outcome.}
    \end{subfigure}
    \caption{Equality corresponding to the gadgetization of a Hadamard gate.}
    \label{fig:gadgetization}
\end{figure}
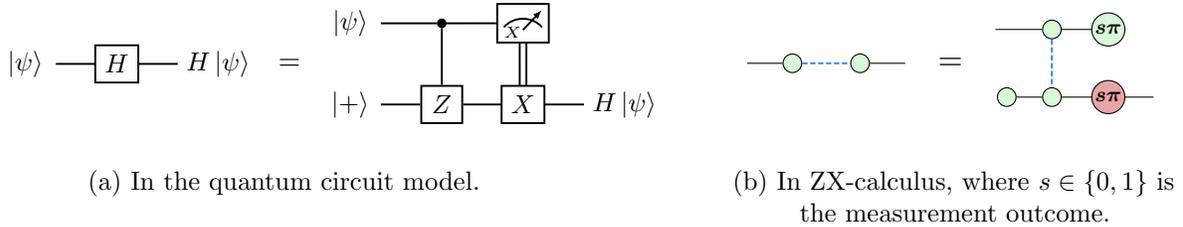

An approach for lowering the number of qubits required by this Hadamard gates gadgetization procedure consists in optimizing the number of Hadamard gates in the initial circuit, as done in References~\cite{de2020fast, vandaele2024optimal}.
The optimization process then consists in optimizing the number of Hadamard gates, gadgetizing them, and optimizing the number of $T$ gates.
To further lower the overhead caused by this Hadamard gates gadgetization stratagem, we propose to append a final step to this process, called Hadamard gates degadgetization, which consists in reversing, as much as possible, the gadgetization of Hadamard gates.
We will show that if certain conditions are satisfied, then it is possible to undo the gadgetization of a Hadamard gate and reduce the number of auxiliary qubits by one without increasing the $T$-count.
To the best of our knowledge, this method for optimizing the number of qubits has not been previously considered.

The second approach we propose consists in using the ZX-calculus as an intermediate language to establish a connection between our qubit-count optimization problem and well-known graph-theoretical problems.
The process then consists in translating a given quantum circuit into a ZX-diagram, performing some transformation on this ZX-diagram and translating it back into a quantum circuit with an optimized number of qubits.
This approach naturally includes various methods for optimizing the number of qubits, such as the degadgetization of Hadamard gates, as described above, or qubit reuse with mid-circuit measurements and resets.
Qubit reuse consists in finding qubits that can be measured before the end of the circuit, resetting them, and reusing them for the rest of the computation.
This approach has already been considered, in particular for near-term quantum computers~\cite{ding2020square, hua2023caqr, decross2023qubit, brandhofer2023optimal, sadeghi2022quantum}.
Our approach differs from these works by also considering both the insertion of new measurements and the removal of some unnecessary measurements (as done for example with the degadgetization of a Hadamard gate) within the circuit to optimize the number of qubits.
These deeper modifications of the circuit are facilitated by the simple and intuitive rules of the ZX-calculus.
For example, the gadgetization of a Hadamard gate, as depicted in Figure~\ref{fig:gadgetization}, can easily be proved using fundamental rules of the ZX-calculus.

Besides its usage as an intermediate representation for performing optimization in quantum circuits, the ZX-calculus can also be used as a direct representation for other computational models.
The Pauli Fusion computational model works natively with the ZX-calculus as it has been designed to reflect its basic structure~\cite{de2019pauli}.
Moreover, it has been shown that the generators of the ZX-calculus are closely related to the basic operations of lattice surgery in the surface code~\cite{de2020zx}.
Thus, the Pauli Fusion computational model can be used to represent lattice surgery operations via ZX-diagrams.
Our approach can then be used to optimize the number of logical qubits required to implement a given ZX-diagram using lattice surgery operations.
However, altering a sequence of lattice surgery operations or modifying a quantum circuit by inserting new measurements can introduce Pauli errors correlated with the measurements outcomes.
A method for correcting these errors is then required.
This can for example be done by relying on the procedure used in Reference~\cite{de2019pauli} to find a Pauli Fusion flow, which provides a correction strategy for deterministically realizing a sequence of lattice surgery operations.

The rest of this paper is organized as follows.
In Section~\ref{sec:preliminaries}, we introduce the background by providing a brief introduction to the ZX-calculus and how it can act as a description for lattice surgery operations.
In Section~\ref{sec:hadamard_gates_degadgetization}, we present an algorithm for degadgetizing a maximal number of Hadamard gates in a given quantum circuit, and we prove the NP-hardness of the problem.
Then, in Section~\ref{sec:qubits_opt_lattice_surgery}, we present algorithms for optimizing the number of qubits required to implement a given ZX-diagram using lattice surgery operations.
We demonstrate how these results can be used to optimize the number of qubits in quantum circuits in Section~\ref{sec:qubits_opt_quantum_circuits}.
Benchmarks are provided in Section~\ref{sec:bench} to evaluate the performances of our algorithms on a library of quantum circuits with an optimized number of non-Clifford gates.
Finally, in Section~\ref{sec:perspectives}, we propose various perspectives for further research work on this problem based on our results.

\section{Preliminaries}\label{sec:preliminaries}

In this section we present the concepts required for the understanding of the results described in this paper.
We first give a brief introduction to the ZX-calculus in Subsection~\ref{sub:zx_calculus}, with a focus on the ZX-calculus rules on which we will rely to optimize the number of qubits.
Then, in Subsection~\ref{sub:lattice_surgery}, we present the logical operations performed by the different lattice surgery procedures and how they can be described by the ZX-calculus.

\subsection{ZX-calculus}\label{sub:zx_calculus}

The ZX-calculus is a diagrammatic language composed of wires and nodes.
The nodes are either green (\begin{tikzpicture}
	\begin{pgfonlayer}{nodelayer}
		\node [style=Z phase dot] (0) at (0, 2) {$\alpha$};
	\end{pgfonlayer}
\end{tikzpicture}
) or red (\begin{tikzpicture}
	\begin{pgfonlayer}{nodelayer}
		\node [style=X phase dot] (0) at (0, 2) {$\alpha$};
	\end{pgfonlayer}
\end{tikzpicture}
), and they have an associated angle, also referred to as the phase and here denoted by $\alpha$, which can be omitted when it is equal to 0.
A spider is a green or red node that can have incident wires, referred to as input wires of the spider if they are attached to the left side of the node or output wires of the spider if they are attached to the right side of the node.
It represents the following linear maps:
\begin{align}
    \begin{tikzpicture}
	\begin{pgfonlayer}{nodelayer}
		\node [style=Label black] (11) at (-2.25, 1) {$n$};
		\node [style=Label black] (12) at (2.25, 1) {$m$};
		\node [style=none] (14) at (1.25, 2) {};
		\node [style=none] (15) at (1.25, 0) {};
		\node [style=none] (16) at (1.25, 1.5) {};
		\node [style=none, rotate=90] (17) at (1, 0.75) {...};
		\node [style=none] (18) at (-1.25, 1.5) {};
		\node [style=Z phase dot] (19) at (0, 1) {$\alpha$};
		\node [style=none] (20) at (-1.25, 2) {};
		\node [style=none] (21) at (-1.25, 0) {};
		\node [style=none, rotate=90] (22) at (-1, 0.75) {...};
	\end{pgfonlayer}
	\begin{pgfonlayer}{edgelayer}
		\draw [style=semithick, decorate, decoration={brace,raise=4pt,amplitude=5pt}] (14.north) to (15.south);
		\draw [style=none, in=-141, out=0, looseness=0.75] (21.center) to (19);
		\draw [style=none, in=0, out=158, looseness=0.75] (19) to (18.center);
		\draw [style=none, in=141, out=0, looseness=0.75] (20.center) to (19);
		\draw [style=semithick, decorate, decoration={brace,raise=4pt,amplitude=5pt}] (21.south) to (20.north);
		\draw [style=none, in=180, out=-39, looseness=0.75] (19) to (15.center);
		\draw [style=none, in=180, out=22, looseness=0.75] (19) to (16.center);
		\draw [style=none, in=180, out=39, looseness=0.75] (19) to (14.center);
	\end{pgfonlayer}
\end{tikzpicture}
    &:= \lvert 0 \rangle^{\otimes m} \langle 0 \rvert^{\otimes n} + e^{i\alpha} \lvert 1 \rangle^{\otimes m} \langle 1 \rvert^{\otimes n}\\
    \begin{tikzpicture}
	\begin{pgfonlayer}{nodelayer}
		\node [style=Label black] (11) at (-2.25, 1) {$n$};
		\node [style=Label black] (12) at (2.25, 1) {$m$};
		\node [style=none] (14) at (1.25, 2) {};
		\node [style=none] (15) at (1.25, 0) {};
		\node [style=none] (16) at (1.25, 1.5) {};
		\node [style=none, rotate=90] (17) at (1, 0.75) {...};
		\node [style=none] (18) at (-1.25, 1.5) {};
		\node [style=X phase dot] (19) at (0, 1) {$\alpha$};
		\node [style=none] (20) at (-1.25, 2) {};
		\node [style=none] (21) at (-1.25, 0) {};
		\node [style=none, rotate=90] (22) at (-1, 0.75) {...};
	\end{pgfonlayer}
	\begin{pgfonlayer}{edgelayer}
		\draw [style=semithick, decorate, decoration={brace,raise=4pt,amplitude=5pt}] (14.north) to (15.south);
		\draw [style=none, in=-141, out=0, looseness=0.75] (21.center) to (19);
		\draw [style=none, in=0, out=158, looseness=0.75] (19) to (18.center);
		\draw [style=none, in=141, out=0, looseness=0.75] (20.center) to (19);
		\draw [style=semithick, decorate, decoration={brace,raise=4pt,amplitude=5pt}] (21.south) to (20.north);
		\draw [style=none, in=180, out=-39, looseness=0.75] (19) to (15.center);
		\draw [style=none, in=180, out=22, looseness=0.75] (19) to (16.center);
		\draw [style=none, in=180, out=39, looseness=0.75] (19) to (14.center);
	\end{pgfonlayer}
\end{tikzpicture}
    &:= \lvert + \rangle^{\otimes m} \langle + \rvert^{\otimes n} + e^{i\alpha} \lvert - \rangle^{\otimes m} \langle - \rvert^{\otimes n}
\end{align}
Green node spiders are called Z spiders and red node spiders are called X spiders.

A ZX-diagram is composed of wires and spiders which can be connected by wires.
Vertical wires can be interpreted as either bent to the right or to the left; it doesn't alter the overall interpretation of the diagram.
A wire is said to be an open wire if at least one of its endpoints is not connected to a node.
For a given ZX-diagram $\mathcal{D}$, an open wire is an input wire of $\mathcal{D}$ if it has at least one open endpoint pointing towards the left, and it is an output wire of $\mathcal{D}$ if it has at least one open endpoint pointing towards the right.
For clarity, we will always put the open endpoint of input wires at the far left of the diagram (before all nodes), and the open endpoint of output wires at the far right of the diagram (after all nodes), as commonly done for quantum circuits.

We can easily translate a quantum circuit composed of CNOT and rotations gates into a ZX-diagram by using the following procedure:
\begin{itemize}
    \item Substitute all $Z$ rotations (respectively, $X$ rotations) of angle $\alpha$ with green nodes (respectively, red nodes) having the same angle $\alpha$.
    \item Substitute all $Z$ measurements (respectively, $X$ measurements) by red nodes (respectively, green nodes).
    \item Substitute all qubits initialization in the $\lvert + \rangle$ state (respectively, $\lvert 0 \rangle$ state) by green nodes (respectively, red nodes).
    \item For each CNOT gate, replace the control node with a green node and the target node with a red node.
    \item Keep the same connectivity.
\end{itemize}
For a Clifford$+T$ circuit, the associated ZX-diagram can be constructed by using the translation table represented in Figure~\ref{fig:quantum_circuit_zx_translation}.
As a shorthand notation, we will represent a Hadamard gate by a dashed blue edge:
\begin{equation}
    \begin{tikzpicture}
	\begin{pgfonlayer}{nodelayer}
		\node [style=Z dot] (0) at (-2.75, 2) {};
		\node [style=Z dot] (1) at (-0.75, 2) {};
		\node [style=none] (2) at (1.25, 2) {$:=$};
		\node [style=none] (4) at (6.75, 2.75) {};
		\node [style=Z phase dot] (5) at (3.25, 2) {$\frac{\pi}{2}$};
		\node [style=X phase dot] (6) at (4.5, 2) {$\frac{\pi}{2}$};
		\node [style=Z phase dot] (7) at (5.75, 2) {$\frac{\pi}{2}$};
		\node [style=none] (8) at (0.25, 2.75) {};
		\node [style=none] (9) at (0.25, 1.25) {};
		\node [style=none] (10) at (-3.75, 2.75) {};
		\node [style=none] (11) at (-3.75, 1.25) {};
		\node [style=none, rotate=90] (12) at (0.25, 2) {...};
		\node [style=none, rotate=90] (13) at (-3.75, 2) {...};
		\node [style=none] (14) at (2.25, 2.75) {};
		\node [style=none] (15) at (2.25, 1.25) {};
		\node [style=none] (16) at (6.75, 1.25) {};
		\node [style=none, rotate=90] (17) at (6.75, 2) {...};
		\node [style=none, rotate=90] (18) at (2.25, 2) {...};
	\end{pgfonlayer}
	\begin{pgfonlayer}{edgelayer}
		\draw [style=Hadamard edge] (1) to (0);
		\draw (1) to (8.center);
		\draw (1) to (9.center);
		\draw (11.center) to (0);
		\draw (10.center) to (0);
		\draw (5) to (7);
		\draw (15.center) to (5);
		\draw (14.center) to (5);
		\draw (7) to (16.center);
		\draw (7) to (4.center);
	\end{pgfonlayer}
\end{tikzpicture}
\end{equation}

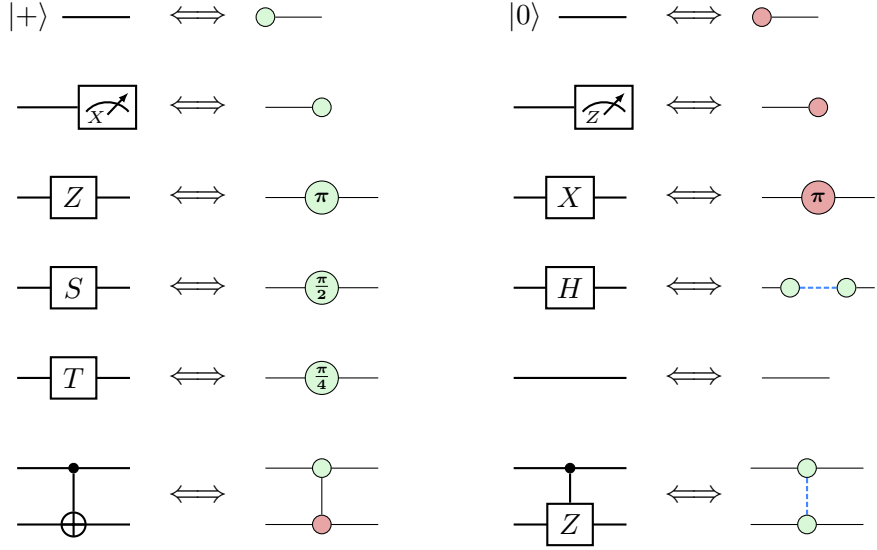
\begin{figure}[t]
    \centering
    \begin{tikzpicture}
	\begin{pgfonlayer}{nodelayer}
		\node [style=none] (1) at (-3, 2) {};
		\node [style=none] (2) at (-1.5, 2) {$\Longleftrightarrow$};
		\node [style=Z dot] (3) at (0, 2) {};
		\node [style=none] (4) at (1.25, 2) {};
		\node [style=Circuit label] (6) at (-5.25, 2) {$\ket{+}$};
		\node [meter, scale=0.9, yshift=0.16mm] (7) at (-3.5, -0.025) {};
		\node [style=none] (8) at (-3.75, -0.225) {\tiny $X$};
		\node [style=none] (9) at (-4.5, 2) {};
		\node [style=none] (17) at (-1.5, 0) {$\Longleftrightarrow$};
		\node [style=Z dot] (18) at (1.25, 0) {};
		\node [style=none] (19) at (0, 0) {};
		\node [style=none] (21) at (-5.5, 0) {};
		\node [style=none] (22) at (-5.5, -2) {};
		\node [style=none] (23) at (-3, -2) {};
		\node [style=Gate] (24) at (-4.25, -2) {$Z$};
		\node [style=none] (25) at (0, -2) {};
		\node [style=none] (26) at (2.5, -2) {};
		\node [style=Z phase dot] (27) at (1.25, -2) {$\pi$};
		\node [style=none] (28) at (-1.5, -2) {$\Longleftrightarrow$};
		\node [style=none] (29) at (8, 2) {};
		\node [style=none] (30) at (9.5, 2) {$\Longleftrightarrow$};
		\node [style=X dot] (31) at (11, 2) {};
		\node [style=none] (32) at (12.25, 2) {};
		\node [style=Circuit label] (33) at (5.75, 2) {$\ket{0}$};
		\node [meter, scale=0.9, yshift=0.16mm] (34) at (7.5, -0.025) {};
		\node [style=none] (35) at (7.25, -0.225) {\tiny $Z$};
		\node [style=none] (36) at (6.5, 2) {};
		\node [style=none] (37) at (9.5, 0) {$\Longleftrightarrow$};
		\node [style=X dot] (38) at (12.25, 0) {};
		\node [style=none] (39) at (11, 0) {};
		\node [style=none] (40) at (5.5, 0) {};
		\node [style=none] (41) at (5.5, -2) {};
		\node [style=none] (42) at (8, -2) {};
		\node [style=Gate] (43) at (6.75, -2) {$X$};
		\node [style=none] (44) at (11, -2) {};
		\node [style=none] (45) at (13.5, -2) {};
		\node [style=X phase dot] (46) at (12.25, -2) {$\pi$};
		\node [style=none] (47) at (9.5, -2) {$\Longleftrightarrow$};
		\node [style=none] (48) at (-5.5, -4) {};
		\node [style=none] (49) at (-3, -4) {};
		\node [style=Gate] (50) at (-4.25, -4) {$S$};
		\node [style=none] (51) at (0, -4) {};
		\node [style=none] (52) at (2.5, -4) {};
		\node [style=Z phase dot] (53) at (1.25, -4) {$\frac{\pi}{2}$};
		\node [style=none] (54) at (-1.5, -4) {$\Longleftrightarrow$};
		\node [style=none] (55) at (-5.5, -6) {};
		\node [style=none] (56) at (-3, -6) {};
		\node [style=Gate] (57) at (-4.25, -6) {$T$};
		\node [style=none] (58) at (0, -6) {};
		\node [style=none] (59) at (2.5, -6) {};
		\node [style=Z phase dot] (60) at (1.25, -6) {$\frac{\pi}{4}$};
		\node [style=none] (61) at (-1.5, -6) {$\Longleftrightarrow$};
		\node [style=none] (62) at (-5.5, -8) {};
		\node [style=none] (63) at (-3, -8) {};
		\node [style=none] (68) at (-1.5, -8.625) {$\Longleftrightarrow$};
		\node [style=Control node] (69) at (-4.25, -8) {};
		\node [style=Target node] (70) at (-4.25, -9.25) {};
		\node [style=none] (71) at (-5.5, -9.25) {};
		\node [style=none] (72) at (-3, -9.25) {};
		\node [style=none] (73) at (0, -8) {};
		\node [style=none] (74) at (2.5, -8) {};
		\node [style=Z dot] (75) at (1.25, -8) {};
		\node [style=X dot] (76) at (1.25, -9.25) {};
		\node [style=none] (77) at (0, -9.25) {};
		\node [style=none] (78) at (2.5, -9.25) {};
		\node [style=none] (79) at (5.5, -8) {};
		\node [style=none] (80) at (8, -8) {};
		\node [style=none] (81) at (9.5, -8.625) {$\Longleftrightarrow$};
		\node [style=Control node] (82) at (6.75, -8) {};
		\node [style=Gate] (83) at (6.75, -9.25) {$Z$};
		\node [style=none] (84) at (5.5, -9.25) {};
		\node [style=none] (85) at (8, -9.25) {};
		\node [style=none] (86) at (10.75, -8) {};
		\node [style=none] (87) at (13.25, -8) {};
		\node [style=Z dot] (88) at (12, -8) {};
		\node [style=Z dot] (89) at (12, -9.25) {};
		\node [style=none] (90) at (10.75, -9.25) {};
		\node [style=none] (91) at (13.25, -9.25) {};
		\node [style=none] (92) at (5.5, -4) {};
		\node [style=none] (93) at (8, -4) {};
		\node [style=Gate] (94) at (6.75, -4) {$H$};
		\node [style=Z dot] (95) at (11.625, -4) {};
		\node [style=Z dot] (96) at (12.875, -4) {};
		\node [style=none] (98) at (9.5, -4) {$\Longleftrightarrow$};
		\node [style=none] (99) at (5.5, -6) {};
		\node [style=none] (100) at (8, -6) {};
		\node [style=none] (102) at (11, -6) {};
		\node [style=none] (103) at (12.5, -6) {};
		\node [style=none] (104) at (9.5, -6) {$\Longleftrightarrow$};
		\node [style=none] (105) at (11, -4) {};
		\node [style=none] (106) at (13.5, -4) {};
	\end{pgfonlayer}
	\begin{pgfonlayer}{edgelayer}
		\draw (4.center) to (3);
		\draw [style=Circuit edge] (9.center) to (1.center);
		\draw (19.center) to (18);
		\draw [style=Circuit edge] (21.center) to (7);
		\draw [style=Circuit edge] (22.center) to (23.center);
		\draw (25.center) to (26.center);
		\draw (32.center) to (31);
		\draw [style=Circuit edge] (36.center) to (29.center);
		\draw (39.center) to (38);
		\draw [style=Circuit edge] (40.center) to (34);
		\draw [style=Circuit edge] (41.center) to (42.center);
		\draw (44.center) to (45.center);
		\draw [style=Circuit edge] (48.center) to (49.center);
		\draw (51.center) to (52.center);
		\draw [style=Circuit edge] (55.center) to (56.center);
		\draw (58.center) to (59.center);
		\draw [style=Circuit edge] (62.center) to (63.center);
		\draw [style=Circuit edge] (69) to (70);
		\draw [style=Circuit edge] (72.center) to (71.center);
		\draw (73.center) to (74.center);
		\draw (75) to (76);
		\draw (78.center) to (77.center);
		\draw [style=Circuit edge] (79.center) to (80.center);
		\draw [style=Circuit edge] (82) to (83);
		\draw [style=Circuit edge] (85.center) to (84.center);
		\draw (86.center) to (87.center);
		\draw [style=Hadamard edge] (88) to (89);
		\draw (91.center) to (90.center);
		\draw [style=Circuit edge] (92.center) to (93.center);
		\draw [style=Hadamard edge] (95) to (96);
		\draw [style=Circuit edge] (99.center) to (100.center);
		\draw (102.center) to (103.center);
		\draw (106.center) to (96);
		\draw (95) to (105.center);
	\end{pgfonlayer}
\end{tikzpicture}
    \caption{Translation table for quantum circuits and ZX-diagrams.}
    \label{fig:quantum_circuit_zx_translation}
\end{figure}

The ZX-calculus has many different rules for transforming a ZX-diagram into an equivalent ZX-diagram.
We now introduce some fundamental rules upon which our results will be built.
The first one is a meta-rule: if two ZX-diagrams have the same open wires and the same set of nodes connected in the same manner, then these two diagrams are equivalent.
This rule is particularly useful as it allows us to treat a family of equivalent ZX-diagrams in a more abstract way by relying on undirected graphs (with additional parameters associated with the vertices).
The second rule, known as the identity rule, states that a wire is equivalent to a spider with no phase and two incident wires:
\begin{equation}
    \begin{tikzpicture}
	\begin{pgfonlayer}{nodelayer}
		\node [style=none] (0) at (-5.25, 0.75) {};
		\node [style=none] (1) at (-3.25, 0.75) {};
		\node [style=none] (2) at (-1.25, 0.75) {};
		\node [style=none] (3) at (0.75, 0.75) {};
		\node [style=none] (4) at (-2.25, 0.75) {$=$};
		\node [style=none] (5) at (2.75, 0.75) {};
		\node [style=none] (6) at (4.75, 0.75) {};
		\node [style=none] (7) at (1.75, 0.75) {$=$};
		\node [style=Z dot] (8) at (-4.25, 0.75) {};
		\node [style=X dot] (9) at (3.75, 0.75) {};
	\end{pgfonlayer}
	\begin{pgfonlayer}{edgelayer}
		\draw (1.center) to (0.center);
		\draw (3.center) to (2.center);
		\draw (6.center) to (5.center);
	\end{pgfonlayer}
\end{tikzpicture}
\end{equation}
Finally, if two Z spiders are connected by at least one wire, then they can be fused together into a single spider, and the resulting phase is the sum of their individual phases:
\begin{equation}
    \begin{tikzpicture}
	\begin{pgfonlayer}{nodelayer}
		\node [style=Z phase dot] (0) at (-2.5, -1.25) {$\beta$};
		\node [style=none] (1) at (-1.25, -0.25) {};
		\node [style=none] (2) at (-4.25, -0.25) {};
		\node [style=none] (3) at (-4.25, -2.25) {};
		\node [style=none] (4) at (-1.25, -2.25) {};
		\node [style=none] (5) at (-1.25, -0.75) {};
		\node [style=none] (6) at (-4.25, -0.75) {};
		\node [style=none, rotate=90] (7) at (-4.5, -1.5) {...};
		\node [style=none, rotate=90] (8) at (-1.5, -1.5) {...};
		\node [style=none] (9) at (-5.25, 1.75) {};
		\node [style=Z phase dot] (10) at (-4, 1.25) {$\alpha$};
		\node [style=none] (11) at (-2.25, 2.25) {};
		\node [style=none] (12) at (-5.25, 2.25) {};
		\node [style=none] (13) at (-2.25, 0.25) {};
		\node [style=none, rotate=90] (14) at (-2, 1) {...};
		\node [style=none] (15) at (-2.25, 1.75) {};
		\node [style=none] (16) at (-5.25, 0.25) {};
		\node [style=none, rotate=90] (17) at (-5, 1) {...};
		\node [style=none] (18) at (-1.25, 0.25) {};
		\node [style=none] (19) at (-1.25, 2.25) {};
		\node [style=none] (20) at (-1.25, 1.75) {};
		\node [style=none] (21) at (-5.25, -0.75) {};
		\node [style=none] (22) at (-5.25, -2.25) {};
		\node [style=none] (23) at (-5.25, -0.25) {};
		\node [style=none] (24) at (0.5, 0) {$=$};
		\node [style=none, rotate=90] (25) at (-3.25, 0) {...};
		\node [style=none] (26) at (3.25, 1.5) {};
		\node [style=none] (27) at (7.25, 1) {};
		\node [style=none, rotate=90] (28) at (3.75, -0.25) {...};
		\node [style=none] (29) at (7.25, -1.5) {};
		\node [style=none, rotate=90] (30) at (7, -0.25) {...};
		\node [style=none] (31) at (3.25, 1) {};
		\node [style=Z phase dot] (32) at (5.25, 0) {$\ \alpha\!+\!\beta\ $};
		\node [style=none] (33) at (7.25, 1.5) {};
		\node [style=none] (34) at (3.25, -1.5) {};
		\node [style=Label black] (35) at (-6.25, 1.25) {$n$};
		\node [style=Label black] (36) at (-6.25, -1.25) {$p$};
		\node [style=Label black] (37) at (-0.2, 1.25) {$m$};
		\node [style=Label black] (38) at (-0.25, -1.25) {$q$};
		\node [style=Label black] (39) at (1.825, 0) {$n\!+\!p$};
		\node [style=Label black] (40) at (8.675, 0) {$m\!+\!q$};
	\end{pgfonlayer}
	\begin{pgfonlayer}{edgelayer}
		\draw [style=none, in=-141, out=0, looseness=0.75] (3.center) to (0);
		\draw [style=none, in=180, out=-39, looseness=0.75] (0) to (4.center);
		\draw [style=none, in=180, out=22, looseness=0.75] (0) to (5.center);
		\draw [style=none, in=180, out=39, looseness=0.75] (0) to (1.center);
		\draw [style=none, in=0, out=180] (0) to (6.center);
		\draw [style=none, in=150, out=0] (2.center) to (0);
		\draw [style=none, in=-141, out=0, looseness=0.75] (16.center) to (10);
		\draw [style=none, in=180, out=-15] (10) to (13.center);
		\draw [style=none, in=180, out=22] (10) to (15.center);
		\draw [style=none, in=180, out=39, looseness=0.75] (10) to (11.center);
		\draw [style=none, in=0, out=158, looseness=0.75] (10) to (9.center);
		\draw [style=none, in=141, out=0, looseness=0.75] (12.center) to (10);
		\draw [style=none, bend left=15] (10) to (0);
		\draw [style=none] (11.center) to (19.center);
		\draw [style=none] (15.center) to (20.center);
		\draw [style=none] (13.center) to (18.center);
		\draw [style=none] (23.center) to (2.center);
		\draw [style=none] (21.center) to (6.center);
		\draw [style=none] (22.center) to (3.center);
		\draw [style=none, bend left=15] (0) to (10);
		\draw [style=none, in=-120, out=0] (34.center) to (32);
		\draw [style=none, in=180, out=-60] (32) to (29.center);
		\draw [style=none, in=180, out=45, looseness=0.75] (32) to (27.center);
		\draw [style=none, in=180, out=60] (32) to (33.center);
		\draw [style=none, in=0, out=135, looseness=0.75] (32) to (31.center);
		\draw [style=none, in=120, out=0] (26.center) to (32);
		\draw [style=semithick, decorate, decoration={brace,raise=4pt,amplitude=5pt}] (16.south) to (12.north);
		\draw [style=semithick, decorate, decoration={brace,raise=4pt,amplitude=5pt}] (22.south) to (23.north);
		\draw [style=semithick, decorate, decoration={brace,raise=4pt,amplitude=5pt}] (1.north) to (4.south);
		\draw [style=semithick, decorate, decoration={brace,raise=4pt,amplitude=5pt}] (19.north) to (18.south);
		\draw [style=semithick, decorate, decoration={brace,raise=4pt,amplitude=5pt}] (34.south) to (26.north);
		\draw [style=semithick, decorate, decoration={brace,raise=4pt,amplitude=5pt}] (33.north) to (29.south);
	\end{pgfonlayer}
\end{tikzpicture}
\end{equation}
This spider fusion rule also holds for X spiders, as every rule of the ZX-calculus remains true when swapping Z and X spiders: 
\begin{equation}
    \begin{tikzpicture}
	\begin{pgfonlayer}{nodelayer}
		\node [style=X phase dot] (0) at (-2.5, -1.25) {$\beta$};
		\node [style=none] (1) at (-1.25, -0.25) {};
		\node [style=none] (2) at (-4.25, -0.25) {};
		\node [style=none] (3) at (-4.25, -2.25) {};
		\node [style=none] (4) at (-1.25, -2.25) {};
		\node [style=none] (5) at (-1.25, -0.75) {};
		\node [style=none] (6) at (-4.25, -0.75) {};
		\node [style=none, rotate=90] (7) at (-4.5, -1.5) {...};
		\node [style=none, rotate=90] (8) at (-1.5, -1.5) {...};
		\node [style=none] (9) at (-5.25, 1.75) {};
		\node [style=X phase dot] (10) at (-4, 1.25) {$\alpha$};
		\node [style=none] (11) at (-2.25, 2.25) {};
		\node [style=none] (12) at (-5.25, 2.25) {};
		\node [style=none] (13) at (-2.25, 0.25) {};
		\node [style=none, rotate=90] (14) at (-2, 1) {...};
		\node [style=none] (15) at (-2.25, 1.75) {};
		\node [style=none] (16) at (-5.25, 0.25) {};
		\node [style=none, rotate=90] (17) at (-5, 1) {...};
		\node [style=none] (18) at (-1.25, 0.25) {};
		\node [style=none] (19) at (-1.25, 2.25) {};
		\node [style=none] (20) at (-1.25, 1.75) {};
		\node [style=none] (21) at (-5.25, -0.75) {};
		\node [style=none] (22) at (-5.25, -2.25) {};
		\node [style=none] (23) at (-5.25, -0.25) {};
		\node [style=none] (24) at (0.5, 0) {$=$};
		\node [style=none, rotate=90] (25) at (-3.25, 0) {...};
		\node [style=none] (26) at (3.25, 1.5) {};
		\node [style=none] (27) at (7.25, 1) {};
		\node [style=none, rotate=90] (28) at (3.75, -0.25) {...};
		\node [style=none] (29) at (7.25, -1.5) {};
		\node [style=none, rotate=90] (30) at (7, -0.25) {...};
		\node [style=none] (31) at (3.25, 1) {};
		\node [style=X phase dot] (32) at (5.25, 0) {$\ \alpha\!+\!\beta\ $};
		\node [style=none] (33) at (7.25, 1.5) {};
		\node [style=none] (34) at (3.25, -1.5) {};
		\node [style=Label black] (35) at (-6.25, 1.25) {$n$};
		\node [style=Label black] (36) at (-6.25, -1.25) {$p$};
		\node [style=Label black] (37) at (-0.2, 1.25) {$m$};
		\node [style=Label black] (38) at (-0.25, -1.25) {$q$};
		\node [style=Label black] (39) at (1.825, 0) {$n\!+\!p$};
		\node [style=Label black] (40) at (8.675, 0) {$m\!+\!q$};
	\end{pgfonlayer}
	\begin{pgfonlayer}{edgelayer}
		\draw [style=none, in=-141, out=0, looseness=0.75] (3.center) to (0);
		\draw [style=none, in=180, out=-39, looseness=0.75] (0) to (4.center);
		\draw [style=none, in=180, out=22, looseness=0.75] (0) to (5.center);
		\draw [style=none, in=180, out=39, looseness=0.75] (0) to (1.center);
		\draw [style=none, in=0, out=180] (0) to (6.center);
		\draw [style=none, in=150, out=0] (2.center) to (0);
		\draw [style=none, in=-141, out=0, looseness=0.75] (16.center) to (10);
		\draw [style=none, in=180, out=-15] (10) to (13.center);
		\draw [style=none, in=180, out=22] (10) to (15.center);
		\draw [style=none, in=180, out=39, looseness=0.75] (10) to (11.center);
		\draw [style=none, in=0, out=158, looseness=0.75] (10) to (9.center);
		\draw [style=none, in=141, out=0, looseness=0.75] (12.center) to (10);
		\draw [style=none, bend left=15] (10) to (0);
		\draw [style=none] (11.center) to (19.center);
		\draw [style=none] (15.center) to (20.center);
		\draw [style=none] (13.center) to (18.center);
		\draw [style=none] (23.center) to (2.center);
		\draw [style=none] (21.center) to (6.center);
		\draw [style=none] (22.center) to (3.center);
		\draw [style=none, bend left=15] (0) to (10);
		\draw [style=none, in=-120, out=0] (34.center) to (32);
		\draw [style=none, in=180, out=-60] (32) to (29.center);
		\draw [style=none, in=180, out=45, looseness=0.75] (32) to (27.center);
		\draw [style=none, in=180, out=60] (32) to (33.center);
		\draw [style=none, in=0, out=135, looseness=0.75] (32) to (31.center);
		\draw [style=none, in=120, out=0] (26.center) to (32);
		\draw [style=semithick, decorate, decoration={brace,raise=4pt,amplitude=5pt}] (16.south) to (12.north);
		\draw [style=semithick, decorate, decoration={brace,raise=4pt,amplitude=5pt}] (22.south) to (23.north);
		\draw [style=semithick, decorate, decoration={brace,raise=4pt,amplitude=5pt}] (1.north) to (4.south);
		\draw [style=semithick, decorate, decoration={brace,raise=4pt,amplitude=5pt}] (19.north) to (18.south);
		\draw [style=semithick, decorate, decoration={brace,raise=4pt,amplitude=5pt}] (34.south) to (26.north);
		\draw [style=semithick, decorate, decoration={brace,raise=4pt,amplitude=5pt}] (33.north) to (29.south);
	\end{pgfonlayer}
\end{tikzpicture}
\end{equation}

\subsection{Lattice surgery}\label{sub:lattice_surgery}

Lattice surgery is a measurement-based method for performing fault-tolerant logical operations, originally developed for the surface code~\cite{horsman2012surface}.
The fundamental lattice surgery operations are the split and merge operations.
The split operation consists in splitting a logical qubit into several logical qubits, whereas the merge operation consists in merging several logical qubits into a single logical qubit.
For the surface code lattice surgery, the split and merge operations can be divided into two types: smooth or rough.
This distinction is based on the type of logical operator on which the operation acts non-trivially.
It has been shown that the ZX-calculus is a language for the surface code lattice surgery~\cite{de2020zx}, where the wires of a ZX-diagram are associated with logical qubits and the spiders are representing the lattice surgery operations acting on these logical qubits.
The red spiders are representing rough operations, while the green spiders are representing smooth operations.
The split operation is represented by a spider having one input wire and multiple output wires.
For example, the smooth and rough splits of a single logical qubit into two logical qubits are represented by a green and red spider respectively, with one input wire and two output wires:
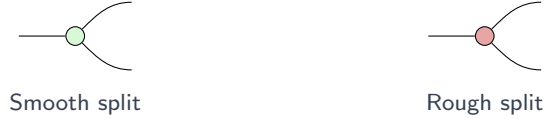
\begin{center}
    \begin{tikzpicture}
	\begin{pgfonlayer}{nodelayer}
		\node [style=none] (0) at (1.25, 2.75) {};
		\node [style=none] (1) at (1.25, 1.25) {};
		\node [style=Z dot] (2) at (0, 2) {};
		\node [style=none] (3) at (-1.25, 2) {};
		\node [style=Text node] (4) at (0, 0.5) {Smooth split};
	\end{pgfonlayer}
	\begin{pgfonlayer}{edgelayer}
		\draw (3.center) to (2);
		\draw [in=180, out=45] (2) to (0.center);
		\draw [in=180, out=-45] (2) to (1.center);
	\end{pgfonlayer}
\end{tikzpicture}
    \qquad\qquad\qquad
    \begin{tikzpicture}
	\begin{pgfonlayer}{nodelayer}
		\node [style=none] (0) at (1.25, 2.75) {};
		\node [style=none] (1) at (1.25, 1.25) {};
		\node [style=X dot] (2) at (0, 2) {};
		\node [style=none] (3) at (-1.25, 2) {};
		\node [style=Text node] (4) at (0, 0.5) {Rough split};
	\end{pgfonlayer}
	\begin{pgfonlayer}{edgelayer}
		\draw (3.center) to (2);
		\draw [in=180, out=45] (2) to (0.center);
		\draw [in=180, out=-45] (2) to (1.center);
	\end{pgfonlayer}
\end{tikzpicture}
\end{center}
The merge operation can be seen as the converse of the split operation.
It is represented by a spider having multiple input wires and one output wire.
However, unlike the split operation, the merge operation is non-deterministic.
Pauli errors, which are correlated to the outcome of the measurements involved in realizing this operation, may be introduced on all but one of the input wires.
For example, the smooth and rough merges of two logical qubits are represented by a green and red spider respectively, with two input wires and one output wire, and where a Pauli error may occur on one of the input wires:
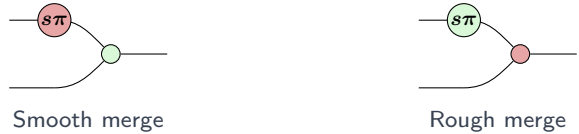
\begin{center}
    \begin{tikzpicture}
	\begin{pgfonlayer}{nodelayer}
		\node [style=Text node] (4) at (-0.5, 0.25) {Smooth merge};
		\node [style=none] (8) at (-1.25, 1) {};
		\node [style=X phase dot] (9) at (-1.25, 2.5) {$s\pi$};
		\node [style=Z dot] (10) at (0, 1.75) {};
		\node [style=none] (11) at (1.25, 1.75) {};
		\node [style=none] (13) at (-2.25, 1) {};
		\node [style=none] (14) at (-2.25, 2.5) {};
	\end{pgfonlayer}
	\begin{pgfonlayer}{edgelayer}
		\draw (11.center) to (10);
		\draw [in=0, out=-135] (10) to (8.center);
		\draw (13.center) to (8.center);
		\draw [in=360, out=135] (10) to (9);
		\draw (14.center) to (9);
	\end{pgfonlayer}
\end{tikzpicture}
    \qquad\qquad\qquad
    \begin{tikzpicture}
	\begin{pgfonlayer}{nodelayer}
		\node [style=Text node] (4) at (-0.5, 0.25) {Rough merge};
		\node [style=none] (8) at (-1.25, 1) {};
		\node [style=Z phase dot] (9) at (-1.25, 2.5) {$s\pi$};
		\node [style=X dot] (10) at (0, 1.75) {};
		\node [style=none] (11) at (1.25, 1.75) {};
		\node [style=none] (13) at (-2.25, 1) {};
		\node [style=none] (14) at (-2.25, 2.5) {};
	\end{pgfonlayer}
	\begin{pgfonlayer}{edgelayer}
		\draw (11.center) to (10);
		\draw [in=0, out=-135] (10) to (8.center);
		\draw (13.center) to (8.center);
		\draw [in=360, out=135] (10) to (9);
		\draw (14.center) to (9);
	\end{pgfonlayer}
\end{tikzpicture}
\end{center}
where $s \in \{0, 1\}$.
The value of $s$ is known only once the merge has been realized, which makes the merge operation non-deterministic by nature.
Yet, some sequences of lattice surgery operations can be realize deterministically even if they include some non-deterministic merge operations.
To do so, a correction must be found for each of the potential Pauli errors.
A correction consists in modifying the subsequent operations to cancel out the effect of the Pauli error.
For instance, the CNOT gate can be implemented using lattice surgery operations as follows:
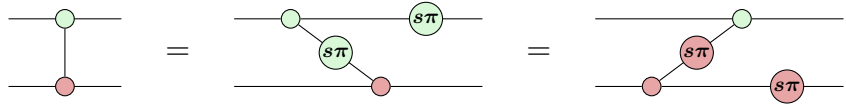
\begin{center}
    \begin{tikzpicture}
	\begin{pgfonlayer}{nodelayer}
		\node [style=Z dot] (13) at (0, 6.25) {};
		\node [style=none] (14) at (-1.25, 6.25) {};
		\node [style=Z phase dot] (16) at (1, 5.5) {$s\pi$};
		\node [style=X dot] (17) at (2, 4.75) {};
		\node [style=none] (18) at (4.25, 4.75) {};
		\node [style=none] (19) at (4.25, 6.25) {};
		\node [style=none] (20) at (-1.25, 4.75) {};
		\node [style=Z phase dot] (21) at (3, 6.25) {$s\pi$};
		\node [style=none] (22) at (-6.25, 6.25) {};
		\node [style=none] (23) at (-3.75, 6.25) {};
		\node [style=none] (24) at (-6.25, 4.75) {};
		\node [style=none] (25) at (-3.75, 4.75) {};
		\node [style=Z dot] (26) at (-5, 6.25) {};
		\node [style=X dot] (27) at (-5, 4.75) {};
		\node [style=none] (28) at (-2.5, 5.5) {$=$};
		\node [style=Z dot] (29) at (10, 6.25) {};
		\node [style=none] (30) at (6.75, 6.25) {};
		\node [style=X phase dot] (31) at (9, 5.5) {$s\pi$};
		\node [style=X dot] (32) at (8, 4.75) {};
		\node [style=none] (33) at (12.25, 4.75) {};
		\node [style=none] (34) at (12.25, 6.25) {};
		\node [style=none] (35) at (6.75, 4.75) {};
		\node [style=X phase dot] (36) at (11, 4.75) {$s\pi$};
		\node [style=none] (37) at (5.5, 5.5) {$=$};
	\end{pgfonlayer}
	\begin{pgfonlayer}{edgelayer}
		\draw (14.center) to (13);
		\draw (18.center) to (17);
		\draw (17) to (16);
		\draw (13) to (16);
		\draw (17) to (20.center);
		\draw (13) to (19.center);
		\draw (24.center) to (25.center);
		\draw (22.center) to (23.center);
		\draw (27) to (26);
		\draw (30.center) to (29);
		\draw (33.center) to (32);
		\draw (32) to (31);
		\draw (29) to (31);
		\draw (32) to (35.center);
		\draw (29) to (34.center);
	\end{pgfonlayer}
\end{tikzpicture}
\end{center}
where $s \in \{0, 1\}$.
An algorithm for finding a correction strategy, called Pauli Fusion flow, for a given ZX-diagram representing lattice surgery operations was proposed in Reference~\cite{de2019pauli}.
In this work, we will assume that any Pauli error has a correction that can be realized on the output wires of the ZX-diagram.
This allows us to freely move the spiders within the ZX-diagram while guaranteeing a deterministic implementation of the sequence of lattice surgery operations it represents.
If this assumption doesn't hold for a given ZX-diagram which can be implemented deterministically, then we can always partition it into several ZX-diagrams, each of which satisfies this property.
This partitioning of a given ZX-diagram can for example be done by using the algorithm proposed in Reference~\cite{de2019pauli}.

In Section~\ref{sec:qubits_opt_lattice_surgery}, we introduce novel algorithms for optimizing the number of qubits required to implement a given ZX-diagram using lattice surgery operations, without connectivity constraints.
In particular, we demonstrate how we can transform a given ZX-diagram $\mathcal{D}$ into another equivalent ZX-diagram $\mathcal{D}'$ by using only the ZX-calculus rules presented in Subsection~\ref{sub:zx_calculus}, and such that $\mathcal{D}'$ can be implemented using lattice surgery operations with a minimal number of logical qubits.
We also prove the NP-hardness of this problem, and show how it relates to well-known graph-theoretical problems.
Then, in Section~\ref{sec:qubits_opt_quantum_circuits}, we demonstrate how these results can be used to optimize the number of qubits in a given quantum circuit.

\section{Hadamard gates degadgetization}\label{sec:hadamard_gates_degadgetization}

In this section, we present a method for optimizing the number of qubits within a quantum circuit based on the degadgetization of Hadamard gates.
We will consider circuits in which all Hadamard gates have been gadgetized.
For example, that is the case of the circuits in which the number of $T$ gates have been optimized by the algorithms of References~\cite{heyfron2018efficient, de2020fast}.
The operations composing this kind of circuits can be represented by diagonal Pauli rotations, classically controlled Clifford gates and a final Clifford operator.
A diagonal Pauli rotation $R_P(\alpha)$ acting on $n$ qubits is defined as 
\begin{equation}
    R_P(\alpha) = \exp(-i\alpha P/2)
\end{equation}
where $P \in \{I, Z\}^{\otimes n}$ is a Pauli product, $I$ is the identity matrix, and $Z$ is the Pauli $Z$ matrix.
We will use $P_i$ to denote the $i$th Pauli matrix of a Pauli product $P$.
A diagonal Pauli rotation can be represented by a phase gadget in ZX-calculus:
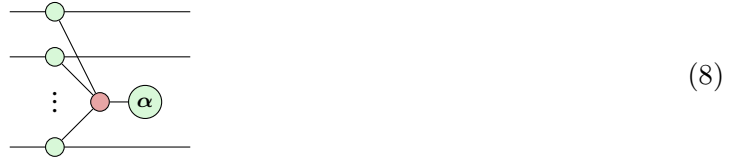
\begin{equation}
    \begin{tikzpicture}
	\begin{pgfonlayer}{nodelayer}
		\node [style=none] (0) at (-2, 3) {};
		\node [style=none] (1) at (-2, 2) {};
		\node [style=none] (2) at (-2, 0) {};
		\node [style=none] (3) at (2, 3) {};
		\node [style=none] (4) at (2, 2) {};
		\node [style=none] (5) at (2, 0) {};
		\node [style=Z dot] (6) at (-1, 3) {};
		\node [style=Z dot] (7) at (-1, 2) {};
		\node [style=Z dot] (8) at (-1, 0) {};
		\node [style=X dot] (9) at (0, 1) {};
		\node [style=Z phase dot] (10) at (1, 1) {$\alpha$};
		\node [style=none, rotate=90] (11) at (-1, 1) {...};
	\end{pgfonlayer}
	\begin{pgfonlayer}{edgelayer}
		\draw (6) to (9);
		\draw (9) to (7);
		\draw (8) to (9);
		\draw (9) to (10);
		\draw (0.center) to (3.center);
		\draw (1.center) to (4.center);
		\draw (2.center) to (5.center);
	\end{pgfonlayer}
\end{tikzpicture}
\end{equation}
We can notice that two phase gadgets necessarily commute.
Thus, if a circuit $C$ is solely composed of phase gadgets, then we can freely rearrange the order of its phase gadgets.

The Hadamard gate gadgetization procedure is presented in Figure~\ref{fig:gadgetization}.
It results in a $CZ_{a, b}$ gate acting on qubits $a$ and $b$ where $a$ is measured in the Pauli $X$ basis and $b$ is prepared in the $\ket{+}$ state.
To undo this Hadamard gate gadgetization, the following two conditions must be satisfied:
\begin{enumerate}
    \item There must be no gate between the $CZ_{a, b}$ gate and the measurement of qubit $a$.
    \item There must be no gate between the initialization of qubit $b$ and the $CZ_{a, b}$ gate.
\end{enumerate}
The problem of degadgetizing a maximal number of Hadamard gates in a quantum circuit composed of phase gadgets then consists in finding an ordering its phase gadgets and $CZ$ gates such that these conditions are satisfied for a maximal number of Hadamard gates gadgetized.
We can encode these conditions into a directed graph $\mathcal{G}$ where each vertex represents a phase gadget or a $CZ$ gate, and where the arcs represent the precedence constraints between the operations that must be respected for all the Hadamard gates to be degadgetized.
Then, the problem consists in finding an ordering of the vertices in $\mathcal{G}$ such that these constraints are satisfied for a maximum number of Hadamard gates gadgetized, which leads to the following problem statement.

\begin{problem}[Hadamard gates degadgetization]
    Let $\mathcal{P} = \{P^{(1)}, \ldots, P^{(m)}\}$ be a set of Pauli products acting on $n$ qubits such that $P^{(i)}_j \in \{I, Z\}$ for all $i, j$, and let $\mathcal{H} = \{(a_1, b_1), \ldots, (a_\ell, b_\ell)\}$ be a set of pairs of integers such that $(a_i, b_i) \in \mathbb{Z}_n^2$.
    Let $\mathcal{G}$ be a directed graph with edge set $V = V_\mathcal{H} \cup V_\mathcal{P}$, where $V_\mathcal{H} = \{h_i \mid (a_i, b_i) \in \mathcal{H} \}$, $V_\mathcal{P} = \{p_i \mid P^{(i)} \in \mathcal{P} \}$, and arc set $A = \{(h_i, h_j) \mid i \neq j, (a_i, b_j) \in \mathcal{H} \} \cup \{(p_i, h_j) \mid P^{(i)}_{a_j} \neq I, P^{(i)} \in \mathcal{P}, (a_j, b_j) \in \mathcal{H} \} \cup \{(h_j, p_i) \mid P^{(i)}_{b_j} \neq I, P^{(i)} \in \mathcal{P}, (a_j, b_j) \in \mathcal{H} \}$.
    Find a subset $X \subseteq V_\mathcal{H}$ of minimal size such that $\mathcal{G}[V \setminus X]$ is an induced directed acyclic graph.
\end{problem}

Let $X$ be a solution to the Hadamard gates degadgetization problem for a quantum circuit $C$ with diagonal Pauli rotations associated with the Pauli products $\mathcal{P} = \{P^{(1)}, \ldots, P^{(m)}\}$ and with $CZ$ gates induced by the gadgetization of Hadamard gates, which are associated with the pairs of qubits $\mathcal{H} = \{(a_1, b_1), \ldots, (a_\ell, b_\ell)\}$.
And let $\mathcal{G}$ be a directed graph with vertex set $V = V_\mathcal{H} \cup V_\mathcal{P}$ constructed from $\mathcal{P}$ and $\mathcal{H}$ as described in the problem statement.
A topological sort of the vertices in the induced directed acyclic graph $\mathcal{G}[V \setminus X]$ provides an ordering of the Pauli rotations and $CZ$ gates of the circuit $C$.
This ordering guarantees that the Hadamard gate degadgetization conditions are satisfied for all the Hadamard gates which are not associated with a vertex in $X$.
Therefore, minimizing the number of vertices in $X$ corresponds to maximizing the number of Hadamard gates degadgetized.

The subset of vertices $X$ is called a feedback vertex set of $\mathcal{G}$, as the induced directed graph $\mathcal{G}[V\setminus X]$ is acyclic.
Finding a minimum feedback vertex set is a well-known NP-hard problem~\cite{karp1972reducibility}.
Note that we are only considering a subset of all possible feedback vertex sets of $\mathcal{G}$, as $X$ must satisfy $X \subseteq V_\mathcal{H}$.
Let $\mathcal{G}'$ be a directed graph with vertex set $V_\mathcal{H}$ and arc set $A' = \{(h_i, h_j) \mid i \neq j, (a_i, b_j) \in \mathcal{H} \} \cup \{(h_i, h_j) \mid (h_i, p_j) \in A, (p_j, h_k) \in A\}$.
Notice that for any path in $\mathcal{G}$ between two vertices in $V_\mathcal{H}$ there exists a corresponding path in $\mathcal{G}'$ that traverses the same vertices in $V_\mathcal{H}$.
Because there is no arcs between the vertices $V_\mathcal{P} \subset V$ of $\mathcal{G}$, it follows that a feedback vertex set $X \subseteq V_\mathcal{H}$ for $\mathcal{G}'$ is also a feedback vertex set for $\mathcal{G}'$.

To summarize, the complete procedure for degadgetizing a maximal number of Hadamard gates in a quantum circuit $C$ is the following:
\begin{enumerate}
    \item Construct the graphs $\mathcal{G}$ and $\mathcal{G}'$ associated with $C$.
    \item Find a minimum feedback vertex set $X$ in $\mathcal{G}'$.
    \item Construct the circuit associated with a topological sort of the vertices of $\mathcal{G}[V \setminus X]$.
\end{enumerate}
An illustrative example of this procedure is provided in Figure~\ref{fig:degadgetization_example}.

\captionsetup[subfigure]{justification=justified, width=0.9\textwidth}
\begin{figure}[t]
    \centering
    \begin{subfigure}{1\textwidth}
        \centering
        \begin{tikzpicture}
	\begin{pgfonlayer}{nodelayer}
		\node [style=Z dot] (10) at (-7, 4) {};
		\node [style=Z dot] (11) at (-7, 0.25) {};
		\node [style=Z phase dot] (23) at (-5, 0.875) {$\alpha_1$};
		\node [style=none] (24) at (-8.25, 0.25) {};
		\node [style=none] (27) at (-8.25, 4) {};
		\node [style=Z dot] (30) at (1, 4) {};
		\node [style=Z dot] (36) at (1, 0.25) {};
		\node [style=Z dot] (60) at (-7, 1.5) {};
		\node [style=Z dot] (61) at (-8, 2.75) {};
		\node [style=Z dot] (62) at (1, 2.75) {};
		\node [style=Z dot] (63) at (-1, 2.75) {};
		\node [style=Z dot] (69) at (-8, -1) {};
		\node [style=Z dot] (72) at (-5, 2.75) {};
		\node [style=Z dot] (73) at (-5, -1) {};
		\node [style=Z phase dot] (75) at (-3, 2.125) {$\alpha_2$};
		\node [style=Z dot] (76) at (-1, 4) {};
		\node [style=Z dot] (77) at (-1, 0.25) {};
		\node [style=Z dot] (78) at (-1, -1) {};
		\node [style=Z dot] (81) at (1, -1) {};
		\node [style=Z phase dot] (86) at (-2, -0.375) {$\alpha_3$};
		\node [style=Z dot] (87) at (-4, -2.25) {};
		\node [style=Z dot] (88) at (-4, 0.25) {};
		\node [style=Z dot] (123) at (-8, -2.25) {};
		\node [style=none] (124) at (1.25, -2.25) {};
		\node [style=Z dot] (125) at (0, -2.25) {};
		\node [style=Z dot] (126) at (0, -1) {};
		\node [style=Z dot] (128) at (-8, 1.5) {};
		\node [style=none] (130) at (1.25, 1.5) {};
		\node [style=Z dot] (131) at (0, 1.5) {};
		\node [style=Z dot] (132) at (0, 2.75) {};
		\node [style=none] (134) at (2.5, 0.875) {$=$};
		\node [style=X dot] (135) at (-6, 0.875) {};
		\node [style=X dot] (136) at (-4, 2.125) {};
		\node [style=X dot] (137) at (-4, 2.125) {};
		\node [style=X dot] (138) at (-3, -0.375) {};
		\node [style=Z dot] (139) at (5, 2.75) {};
		\node [style=Z dot] (140) at (5, 0.25) {};
		\node [style=Z phase dot] (141) at (7, 0.875) {$\alpha_1$};
		\node [style=none] (142) at (3.75, 0.25) {};
		\node [style=none] (143) at (3.75, 2.75) {};
		\node [style=Z dot] (144) at (9, 2.75) {};
		\node [style=Z dot] (145) at (9, 0.25) {};
		\node [style=Z dot] (146) at (5, 1.5) {};
		\node [style=Z dot] (147) at (10.25, 2.75) {};
		\node [style=Z dot] (148) at (15.5, 2.75) {};
		\node [style=Z dot] (149) at (10.25, 2.75) {};
		\node [style=Z dot] (151) at (11.25, 2.75) {};
		\node [style=Z dot] (152) at (11.25, 0.25) {};
		\node [style=Z phase dot] (153) at (13.25, 2.125) {$\alpha_2$};
		\node [style=Z dot] (154) at (9, 2.75) {};
		\node [style=Z dot] (155) at (9, 0.25) {};
		\node [style=Z dot] (156) at (10.25, 0.25) {};
		\node [style=Z dot] (157) at (15.5, 0.25) {};
		\node [style=Z phase dot] (158) at (8, -0.375) {$\alpha_3$};
		\node [style=Z dot] (159) at (6, -1) {};
		\node [style=Z dot] (160) at (6, 0.25) {};
		\node [style=Z dot] (161) at (4, -1) {};
		\node [style=none] (162) at (15.75, -1) {};
		\node [style=Z dot] (163) at (14.5, -1) {};
		\node [style=Z dot] (164) at (14.5, 0.25) {};
		\node [style=Z dot] (165) at (4, 1.5) {};
		\node [style=none] (166) at (15.75, 1.5) {};
		\node [style=Z dot] (167) at (14.5, 1.5) {};
		\node [style=Z dot] (168) at (14.5, 2.75) {};
		\node [style=X dot] (169) at (6, 0.875) {};
		\node [style=X dot] (170) at (12.25, 2.125) {};
		\node [style=X dot] (171) at (12.25, 2.125) {};
		\node [style=X dot] (172) at (7, -0.375) {};
		\node [style=Label node] (177) at (-0.5, 3.375) {$h_1$};
		\node [style=Label node] (178) at (0.5, 2.125) {$h_2$};
		\node [style=Label node] (179) at (-0.5, -0.375) {$h_3$};
		\node [style=Label node] (180) at (0.5, -1.625) {$h_4$};
		\node [style=Label node] (181) at (9.625, 3.25) {$h_1$};
		\node [style=Label node] (182) at (15, 2.175) {$h_2$};
		\node [style=Label node] (183) at (9.625, 0.75) {$h_3$};
		\node [style=Label node] (184) at (15, -0.375) {$h_4$};
	\end{pgfonlayer}
	\begin{pgfonlayer}{edgelayer}
		\draw (62) to (61);
		\draw (24.center) to (36);
		\draw (69) to (78);
		\draw (78) to (81);
		\draw [style=Hadamard edge] (77) to (78);
		\draw [style=Hadamard edge] (76) to (63);
		\draw (27.center) to (30);
		\draw (123) to (124.center);
		\draw [style=Hadamard edge] (126) to (125);
		\draw (128) to (130.center);
		\draw [style=Hadamard edge] (132) to (131);
		\draw (60) to (135);
		\draw (10) to (135);
		\draw (11) to (135);
		\draw (135) to (23);
		\draw (136) to (72);
		\draw (137) to (73);
		\draw [in=180, out=0] (138) to (86);
		\draw (138) to (87);
		\draw (88) to (138);
		\draw (137) to (75);
		\draw (148) to (147);
		\draw (142.center) to (145);
		\draw (156) to (157);
		\draw [style=Hadamard edge] (155) to (156);
		\draw [style=Hadamard edge] (154) to (149);
		\draw (143.center) to (144);
		\draw (161) to (162.center);
		\draw [style=Hadamard edge] (164) to (163);
		\draw (165) to (166.center);
		\draw [style=Hadamard edge] (168) to (167);
		\draw (146) to (169);
		\draw (139) to (169);
		\draw (140) to (169);
		\draw (169) to (141);
		\draw (170) to (151);
		\draw (171) to (152);
		\draw [in=180, out=0] (172) to (158);
		\draw (172) to (159);
		\draw (160) to (172);
		\draw (171) to (153);
	\end{pgfonlayer}
\end{tikzpicture}
        \caption{The initial quantum circuit (represented in ZX-calculus with phase gadgets) in which Hadamard gates have been gadgetized, and an equivalent circuit in which $2$ Hadamard gates have been degadgetized.}
    \end{subfigure}\vspace{0.3cm}
    \begin{subfigure}{1\textwidth}
        \centering
        \begin{tikzpicture}
	\begin{pgfonlayer}{nodelayer}
		\node [style=Graph node] (0) at (-2, 2.5) {$h_1$};
		\node [style=Graph node] (1) at (-2, 1) {$h_2$};
		\node [style=Graph node] (2) at (-2, -0.5) {$h_3$};
		\node [style=Graph node] (3) at (-2, -2) {$h_4$};
		\node [style=Graph node] (4) at (-5, 1.75) {$p_1$};
		\node [style=Graph node] (5) at (-5, 0.25) {$p_2$};
		\node [style=Graph node] (6) at (-5, -1.25) {$p_3$};
		\node [style=Graph node] (7) at (4, 2.5) {$h_1$};
		\node [style=Graph node, fill=lightgray] (8) at (4, 1) {$h_2$};
		\node [style=Graph node] (9) at (4, -0.5) {$h_3$};
		\node [style=Graph node, fill=lightgray] (10) at (4, -2) {$h_4$};
		\node [style=Graph node] (11) at (14, 2.5) {$h_1$};
		\node [style=Graph node] (12) at (14, -0.5) {$h_3$};
		\node [style=Graph node] (13) at (11, 1.75) {$p_1$};
		\node [style=Graph node] (14) at (11, 0.25) {$p_2$};
		\node [style=Graph node] (15) at (11, -1.25) {$p_3$};
		\node [style=none] (16) at (0, 0.25) {};
		\node [style=none] (17) at (2, 0.25) {};
		\node [style=none] (18) at (7, 0.25) {};
		\node [style=none] (19) at (9, 0.25) {};
		\node [style=Text node] (20) at (1, 0.925) {{\tiny Minimum feedback\\[-3pt]vertex set}};
		\node [style=Text node] (21) at (8, 0.925) {{\tiny Directed acyclic\\[-3pt]graph}};
	\end{pgfonlayer}
	\begin{pgfonlayer}{edgelayer}
		\draw [style=Arrow2, in=195, out=15] (4) to (0);
		\draw [style=Arrow2] (0) to (5);
		\draw [style=Arrow2] (5) to (1);
		\draw [style=Arrow2] (6) to (2);
		\draw [style=Arrow2] (3) to (6);
		\draw [style=Arrow2] (2) to (5);
		\draw [style=Arrow2] (5) to (3);
		\draw [style=Arrow2, in=60, out=-60] (7) to (8);
		\draw [style=Arrow2, in=30, out=-30] (7) to (10);
		\draw [style=Arrow2] (1) to (4);
		\draw [style=Arrow2] (4) to (2);
		\draw [style=Arrow2, in=60, out=-60] (8) to (9);
		\draw [style=Arrow2, in=240, out=120] (9) to (8);
		\draw [style=Arrow2, in=60, out=-60] (9) to (10);
		\draw [style=Arrow2, in=240, out=120] (10) to (9);
		\draw [style=Arrow2, in=240, out=120] (8) to (7);
		\draw [style=Arrow2] (13) to (11);
		\draw [style=Arrow2] (11) to (14);
		\draw [style=Arrow2] (15) to (12);
		\draw [style=Arrow2] (12) to (14);
		\draw [style=Arrow2] (13) to (12);
		\draw [style=Arrow] (16.center) to (17.center);
		\draw [style=Arrow] (18.center) to (19.center);
	\end{pgfonlayer}
\end{tikzpicture}
        \caption{Corresponding directed graph where the vertices $p_1$, $p_2$, $p_3$ are associated with the phase gadgets with angle $\alpha_1$, $\alpha_2$, $\alpha_3$ respectively, and the vertices $h_1, h_2, h_3, h_4$ are associated with the $CZ$ gates.
        The minimum feedback vertex set, here corresponding to the set of vertices $\{h_2, h_4\}$, indicates the Hadamard gates that will not be degadgetized.
        A topological sorting of the resulting directed acyclic graph gives the order of the phase gadgets and the Hadamard gates in the optimized circuit.}
    \end{subfigure}
    \caption{Example demonstrating the procedure used for degadgetizing a maximal number of Hadamard gates in a quantum circuit.}
    \label{fig:degadgetization_example}
\end{figure}
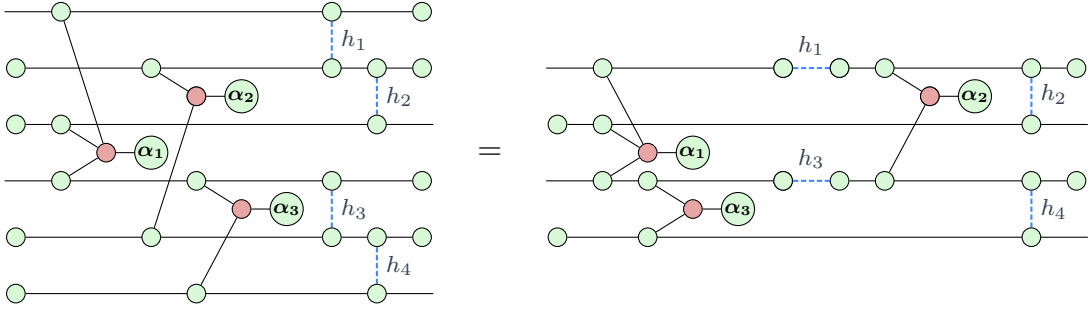
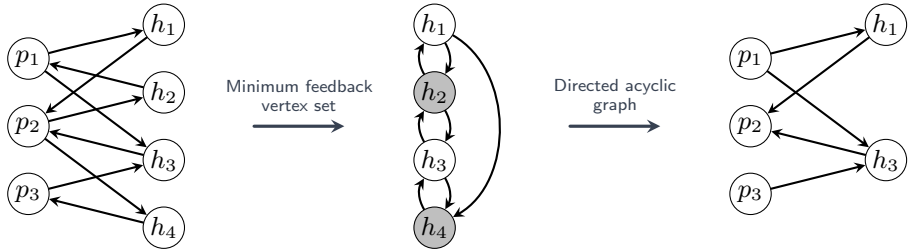
\captionsetup[subfigure]{justification=centering, width=0.9\textwidth}

We now prove the NP-hardness of the Hadamard gates degadgetization problem by reducing the minimum directed feedback vertex set problem to this problem.

\begin{theorem}
    The Hadamard gate degadgetization problem is NP-hard.
\end{theorem}

\begin{proof}
    Let $\mathcal{G}$ be a directed graph with vertex set $V$ and arc set $A$.
    Let $\mathcal{P} = \{P^{(u,v)} \mid P^{(u, v)}_{u + \lvert V \rvert} = P^{(u, v)}_v = Z, P^{(u, v)}_w = I,  (u, v) \in A, w \not \in \{u, v\} \}$ and let $\mathcal{H} = \{(u, u + \lvert V \rvert) \mid u \in V \}$.
    Let $X$ be a solution to the Hadamard gates degadgetization problem for $\mathcal{P}$ and $\mathcal{H}$.
    Then, $X$ is a minimum feedback vertex set of the graph $\mathcal{G}'$ with vertex set $V' = \{h_u \mid u \in V\}$ and arc set $A' = \{(h_u, h_v) \mid P^{(u,v)} \in \mathcal{P} \} = \{(h_u, h_v) \mid (u, v) \in A \}$.
    The directed graph $\mathcal{G}'$ is isomorphic to the directed graph $\mathcal{G}$.
    Since $X$ is a minimum feedback vertex set for $\mathcal{G}'$, it follows that $X$ is also a feedback vertex set for $\mathcal{G}$.
    Thus, the Hadamard gates degadgetization problem is at least as hard as the minimum directed feedback vertex set, which is an NP-hard problem~\cite{karp1972reducibility}.
\end{proof}

\section{Qubit-count optimization for lattice surgery using ZX-calculus}\label{sec:qubits_opt_lattice_surgery}

In this section we present a graph-theoretical approach to the qubit-count optimization problem by relying on the ZX-calculus as a language for lattice surgery.
The proposed methods can be used to optimize the number of logical qubits required to implement the lattice surgery operations represented by a given ZX-diagram, without any connectivity constraints.
These methods can also be used to optimize the number of qubits in a given quantum circuit, as explained in Section~\ref{sec:qubits_opt_quantum_circuits}.
In Subsection~\ref{sub:qubits_opt_vertex_ordering}, we show how we can optimize the number of qubits by rearranging the order of the spiders in the ZX-diagram.
Then, in Subsection~\ref{sub:qubits_opt_fusion_rule}, we show how we can extend this result by incorporating the spider fusion rule to better optimize the number of qubits.

\subsection{Qubit-count optimization via spider ordering}\label{sub:qubits_opt_vertex_ordering}

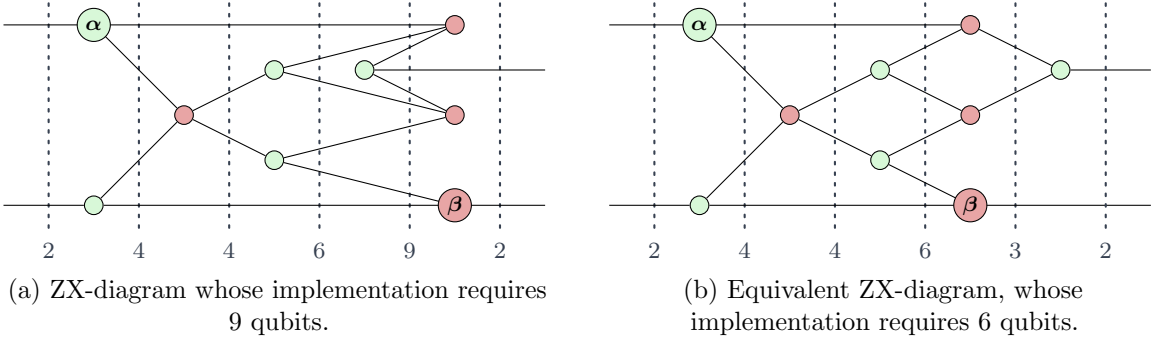
\begin{figure}[t]
    \centering
    \begin{subfigure}{0.48\textwidth}
        \centering
        \begin{tikzpicture}[{every label/.append style}={font=\scriptsize, text={rgb,255: red,55; green,65; blue,81}}]
	\begin{pgfonlayer}{nodelayer}
		\node [style=none] (0) at (-1, 4) {};
		\node [style=none] (1) at (-1, 0) {};
		\node [style=Z phase dot] (2) at (1, 4) {$\alpha$};
		\node [style=Z dot] (3) at (1, 0) {};
		\node [style=X dot] (4) at (9, 2) {};
		\node [style=Z dot] (5) at (5, 3) {};
		\node [style=Z dot] (6) at (7, 3) {};
		\node [style=X dot] (7) at (9, 4) {};
		\node [style=none] (8) at (11, 3) {};
		\node [style=none] (9) at (11, 0) {};
		\node [style=X phase dot] (10) at (9, 0) {$\beta$};
		\node [style=X dot] (11) at (3, 2) {};
		\node [style=Z dot] (12) at (5, 1) {};
		\node [style=none] (13) at (0, 4.5) {};
		\node [style=none] (14) at (0, -0.5) {};
		\node [style=none] (19) at (2, 4.5) {};
		\node [style=none] (20) at (2, -0.5) {};
		\node [style=none] (21) at (4, 4.5) {};
		\node [style=none] (22) at (4, -0.5) {};
		\node [style=none] (23) at (6, 4.5) {};
		\node [style=none] (24) at (6, -0.5) {};
		\node [style=none] (25) at (8, 4.5) {};
		\node [style=none] (26) at (8, -0.5) {};
		\node [style=none] (27) at (10, 4.5) {};
		\node [style=none] (28) at (10, -0.5) {};
		\node [style=Label node] (29) at (0, -1) {2};
		\node [style=Label node] (30) at (2, -1) {4};
		\node [style=Label node] (31) at (4, -1) {4};
		\node [style=Label node] (32) at (6, -1) {6};
		\node [style=Label node] (33) at (8, -1) {9};
		\node [style=Label node] (34) at (10, -1) {2};
	\end{pgfonlayer}
	\begin{pgfonlayer}{edgelayer}
		\draw [style=Dashed grey] (14.center) to (13.center);
		\draw [style=Dashed grey] (20.center) to (19.center);
		\draw [style=Dashed grey] (22.center) to (21.center);
		\draw [style=Dashed grey] (24.center) to (23.center);
		\draw [style=Dashed grey] (26.center) to (25.center);
		\draw [style=Dashed grey] (28.center) to (27.center);
		\draw (5) to (4);
		\draw (5) to (7);
		\draw (7) to (6);
		\draw (4) to (6);
		\draw (2) to (7);
		\draw (1.center) to (3);
		\draw (0.center) to (2);
		\draw (6) to (8.center);
		\draw (3) to (10);
		\draw (11) to (2);
		\draw (10) to (9.center);
		\draw (3) to (11);
		\draw (12) to (11);
		\draw (12) to (4);
		\draw (12) to (10);
		\draw (11) to (5);
	\end{pgfonlayer}
\end{tikzpicture}
        \caption{ZX-diagram whose implementation requires $9$ qubits.}
    \end{subfigure}
    \begin{subfigure}{0.48\textwidth}
        \centering
        \begin{tikzpicture}
	\begin{pgfonlayer}{nodelayer}
		\node [style=none] (0) at (-1, 4) {};
		\node [style=none] (1) at (-1, 0) {};
		\node [style=Z phase dot] (2) at (1, 4) {$\alpha$};
		\node [style=Z dot] (3) at (1, 0) {};
		\node [style=X dot] (4) at (7, 2) {};
		\node [style=Z dot] (5) at (5, 3) {};
		\node [style=Z dot] (6) at (9, 3) {};
		\node [style=X dot] (7) at (7, 4) {};
		\node [style=none] (8) at (11, 3) {};
		\node [style=none] (9) at (11, 0) {};
		\node [style=X phase dot] (10) at (7, 0) {$\beta$};
		\node [style=X dot] (11) at (3, 2) {};
		\node [style=Z dot] (12) at (5, 1) {};
		\node [style=none] (13) at (0, 4.5) {};
		\node [style=none] (14) at (0, -0.5) {};
		\node [style=none] (19) at (2, 4.5) {};
		\node [style=none] (20) at (2, -0.5) {};
		\node [style=none] (21) at (4, 4.5) {};
		\node [style=none] (22) at (4, -0.5) {};
		\node [style=none] (23) at (6, 4.5) {};
		\node [style=none] (24) at (6, -0.5) {};
		\node [style=none] (27) at (8, 4.5) {};
		\node [style=none] (28) at (8, -0.5) {};
		\node [style=none] (29) at (10, 4.5) {};
		\node [style=none] (30) at (10, -0.5) {};
		\node [style=Label node] (31) at (0, -1) {2};
		\node [style=Label node] (32) at (2, -1) {4};
		\node [style=Label node] (33) at (4, -1) {4};
		\node [style=Label node] (34) at (6, -1) {6};
		\node [style=Label node] (35) at (8, -1) {3};
		\node [style=Label node] (36) at (10, -1) {2};
	\end{pgfonlayer}
	\begin{pgfonlayer}{edgelayer}
		\draw [style=Dashed grey] (14.center) to (13.center);
		\draw [style=Dashed grey] (20.center) to (19.center);
		\draw [style=Dashed grey] (22.center) to (21.center);
		\draw [style=Dashed grey] (24.center) to (23.center);
		\draw [style=Dashed grey] (28.center) to (27.center);
		\draw (5) to (4);
		\draw (5) to (7);
		\draw (7) to (6);
		\draw (4) to (6);
		\draw (2) to (7);
		\draw (1.center) to (3);
		\draw (0.center) to (2);
		\draw (6) to (8.center);
		\draw (3) to (10);
		\draw (11) to (2);
		\draw (10) to (9.center);
		\draw (3) to (11);
		\draw (12) to (11);
		\draw (12) to (4);
		\draw (12) to (10);
		\draw (11) to (5);
		\draw [style=Dashed grey] (30.center) to (29.center);
	\end{pgfonlayer}
\end{tikzpicture}
        \caption{Equivalent ZX-diagram, whose implementation requires $6$ qubits.}
    \end{subfigure}
    \caption{Example demonstrating how rearranging the order of spiders in a given ZX-diagrams can reduce the number of logical qubits needed for its implementation.
        The ZX-diagrams are annotated with the required number of logical qubits at each step to implement them with the lattice surgery operations they represent.}
    \label{fig:spider_ordering_example}
\end{figure}

The operation realized by a ZX-diagram is preserved even when its spiders are moved and its wires are bent: only connectivity matters.
For example, the following diagrams all correspond to the CNOT gate, despite representing different sequences of lattice surgery operations:
\begin{equation}
    \begin{tikzpicture}
	\begin{pgfonlayer}{nodelayer}
		\node [style=Z dot] (0) at (-1, 2.75) {};
		\node [style=X dot] (1) at (0.25, 1.25) {};
		\node [style=none] (2) at (-2, 2.75) {};
		\node [style=none] (3) at (-2, 1.25) {};
		\node [style=none] (4) at (1.25, 1.25) {};
		\node [style=none] (5) at (1.25, 2.75) {};
		\node [style=none] (6) at (2.25, 2) {$=$};
		\node [style=none] (13) at (7.5, 2) {$=$};
		\node [style=Z dot] (14) at (15.5, 2.75) {};
		\node [style=X dot] (15) at (15.5, 1.25) {};
		\node [style=none] (16) at (13.5, 2.75) {};
		\node [style=none] (17) at (13.5, 1.25) {};
		\node [style=none] (18) at (16.5, 1.25) {};
		\node [style=none] (19) at (16.5, 2.75) {};
		\node [style=Z dot] (20) at (5.5, 2.75) {};
		\node [style=X dot] (21) at (4.25, 1.25) {};
		\node [style=none] (22) at (3.25, 2.75) {};
		\node [style=none] (23) at (3.25, 1.25) {};
		\node [style=none] (24) at (6.5, 1.25) {};
		\node [style=none] (25) at (6.5, 2.75) {};
		\node [style=none] (26) at (12.5, 2) {$=$};
		\node [style=X dot] (27) at (9.5, 1.25) {};
		\node [style=Z dot] (28) at (9.5, 2.75) {};
		\node [style=none] (29) at (8.5, 2.75) {};
		\node [style=none] (30) at (8.5, 1.25) {};
		\node [style=none] (31) at (11.5, 1.25) {};
		\node [style=none] (32) at (11.5, 2.75) {};
	\end{pgfonlayer}
	\begin{pgfonlayer}{edgelayer}
		\draw (0) to (1);
		\draw (2.center) to (0);
		\draw (3.center) to (1);
		\draw (4.center) to (1);
		\draw (5.center) to (0);
		\draw [in=165, out=-165, looseness=2.25] (14) to (15);
		\draw (16.center) to (14);
		\draw (17.center) to (15);
		\draw (18.center) to (15);
		\draw (19.center) to (14);
		\draw (20) to (21);
		\draw (22.center) to (20);
		\draw (23.center) to (21);
		\draw (24.center) to (21);
		\draw (25.center) to (20);
		\draw [in=-15, out=15, looseness=2.25] (27) to (28);
		\draw (29.center) to (32.center);
		\draw (31.center) to (30.center);
	\end{pgfonlayer}
\end{tikzpicture}
\end{equation}
The number of logical qubits required to implement a given ZX-diagram using lattice surgery operations is equal to the maximum number of wires in any vertical cut of the diagram.
Thus, rearranging the order of the spiders of a ZX-diagram or bending (or unbending) its wires can reduce the number of qubits required to implement the associated sequence of lattice surgery operations.
An example is provided in Figure~\ref{fig:spider_ordering_example}.
Then, our objective is to transform a given ZX-diagram by moving its spiders and bending (or unbending) its wires such that the maximum number of wires over any vertical cut of the diagram is minimized.
To formalize this problem in a graph-theoretical manner, we define the signature of a ZX-diagram as follows.

\begin{definition}[ZX-diagram signature]
    Let $\mathcal{D}$ be a ZX-diagram, and let $\mathcal{G}_\mathcal{D}$ be the signature graph of $\mathcal{D}$ which is an undirected graph such that
    \begin{itemize}
        \item it contains a vertex for each spider of $\mathcal{D}$, for each input wire of $\mathcal{D}$ ($\mathcal{I}$ denotes this set of vertices) and for each output wire of $\mathcal{D}$ ($\mathcal{O}$ denotes this set of vertices),
        \item and it has an edge between the vertices $u$ and $v$ for all spiders $u$ and $v$ connected in $\mathcal{D}$, for all $u \in \mathcal{I}$ and spider $v$ incident to the input wire associated with $u$ in $\mathcal{D}$ and for all $u \in \mathcal{O}$ and spider $v$ incident to the output wire associated with $u$ in $\mathcal{D}$.
    \end{itemize}
    Then, the triple $(\mathcal{G}_\mathcal{D}, \mathcal{I}, \mathcal{O})$ is the signature associated with $\mathcal{D}$.
\end{definition}

And we will use the following definition of vertex ordering for a given graph to incorporate the notion of moving the spiders of a ZX-diagram in our problem.

\begin{definition}[Vertex ordering]
    Let $\mathcal{G}$ be a graph with vertex set $V$.
    A function $f$ is an ordering of the vertices $V$ if and only if $f(u) \in \{1, 2, \ldots, \lvert V \rvert\}$ for all $u \in V$ and $f(u) \neq f(v)$ for all $u,v \in V$ such that $u \neq v$.
\end{definition}

We will see that our qubit-count optimization problem is closely related to the graph-theoretical concept of cutwidth, defined as follows.

\begin{definition}[Cutwidth]
    Let $\mathcal{G}=(V,E)$ be a graph with vertex set $V$ and edge set $E$.
    The cutwidth of a vertex $v \in V$ with respect to a vertex ordering $f$, denoted by $cw_f(v)$, is defined as:
    \begin{equation}
        cw_f(v) = \lvert \{(u, w) \mid f(u) \leq f(v) < f(w), (u, w) \in E \} \rvert.
    \end{equation}
    The cutwidth of $\mathcal{G}$ with respect to $f$, denoted by $cw_f(\mathcal{G})$, is defined as: 
    \begin{equation}
        cw_f(\mathcal{G}) = \max_{v\in V} cw_f(v).
    \end{equation}
    And the cutwidth of $\mathcal{G}$, denoted by $cw(\mathcal{G})$, is defined as the minimum $cw_f(\mathcal{G})$ value over all possible orderings:
    \begin{equation}
        cw(\mathcal{G}) = \min_{f \in S_V} cw_f(\mathcal{G}) = \min_{f \in S_V} \max_{v\in V} cw_f(v)
    \end{equation}
    where $S_V$ is the set of all orderings of the vertices in $V$.
\end{definition}

The cutwidth problem consists in finding the cutwidth $cw(\mathcal{G})$ of a given graph $\mathcal{G}$.
This problem is known to be NP-hard~\cite{gavril1977some}.

We now show how rearranging the spiders of a ZX-diagram to minimize the number of qubits required to implement it using lattice surgery operations relates to the cutwidth problem.
Let $(\mathcal{G}_\mathcal{D}, \mathcal{I}, \mathcal{O})$ be the signature of a ZX-diagram $\mathcal{D}$ and $V$ be the vertices of $\mathcal{G}_\mathcal{D}$.
Choosing an ordering $f$ of the vertices $V$ of $\mathcal{G}_\mathcal{D}$ can be seen as moving the spider of the ZX-diagram $\mathcal{D}$, provided that $f(u) < f(v) < f(w)$ for all $u \in \mathcal{I}$, $v \in V\setminus \{\mathcal{I} \cup \mathcal{O}\}$, and $w \in \mathcal{O}$ (as a spider cannot precede an input wire or succeed an output wire).
Indeed, we can construct a ZX-diagram $\mathcal{D}'$ from $f$ such that $\mathcal{D}'$ contains the same nodes as $\mathcal{D}$, connected in the same manner (and connected to the same input and output wires), but where the spiders of $\mathcal{D}'$ are ordered according to $f$.
Then, it follows that the ZX-diagrams $\mathcal{D}'$ and $\mathcal{D}$ are equivalent, but the number of qubits required to implement $\mathcal{D}'$ in lattice surgery is equal to $cw_f(\mathcal{G_\mathcal{D}})$.
Thus, we want to find an ordering $f$ such that $cw_f(\mathcal{G_\mathcal{D}})$ is minimized, and where $f$ satisfies $f(u) < f(v) < f(w)$ for all $u \in \mathcal{I}$, $v \in V\setminus \{\mathcal{I} \cup \mathcal{O}\}$, and $w \in \mathcal{O}$.
This problem corresponds to the problem of finding an ordering $f$ of the vertices of $\mathcal{G}_\mathcal{D}$ with minimum cutwidth, but where the first and last vertices of the ordering are fixed.
We refer to this specific problem as the fixed-endvertices cutwidth problem.

\begin{problem}[Fixed-endvertices cutwidth]
    Let $\mathcal{G}=(V,E)$ be a graph with vertex set $V$ and edge set $E$ and let $u, w \in V$ such that $u \neq w$.
    Find an ordering $f$ of the vertices $V$ satisfying $f(u) \leq f(v) \leq f(w)$ for all $v \in V$ and such that $cw_f(\mathcal{G})$ is minimized.
\end{problem}

An illustrative example demonstrating how the fixed-endvertices cutwidth problem can be used to minimize the maximum number of wires in any vertical cut of a given ZX-diagram is provided in Figure~\ref{fig:cutwidth_example}.

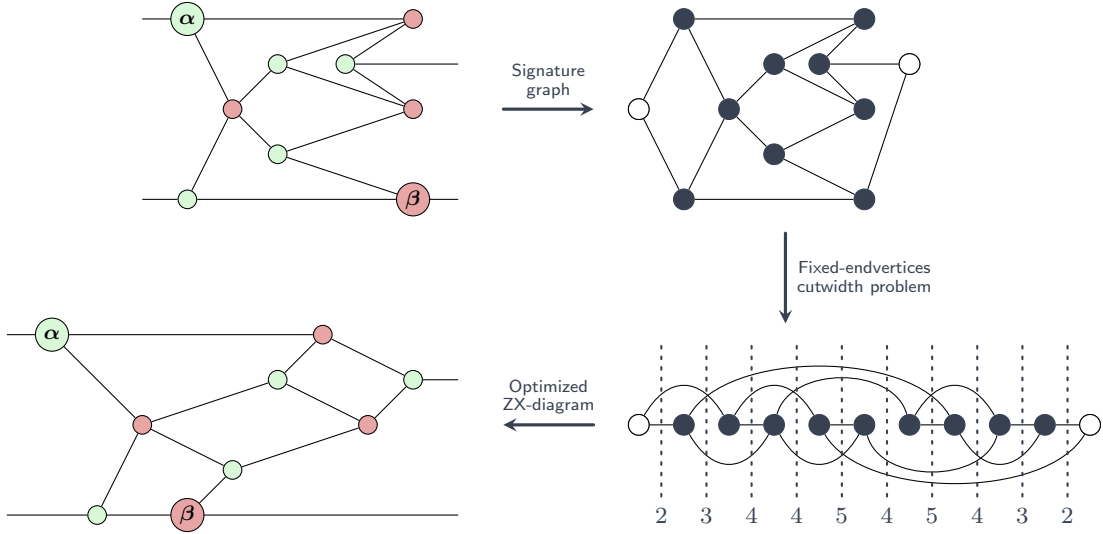
\begin{figure}[t]
    \centering
    \begin{tikzpicture}
	\begin{pgfonlayer}{nodelayer}
		\node [style=none] (0) at (-3, 4) {};
		\node [style=none] (1) at (-3, 0) {};
		\node [style=Z phase dot] (2) at (-2, 4) {$\alpha$};
		\node [style=Z dot] (3) at (-2, 0) {};
		\node [style=X dot] (4) at (3, 2) {};
		\node [style=Z dot] (5) at (0, 3) {};
		\node [style=Z dot] (6) at (1.5, 3) {};
		\node [style=X dot] (7) at (3, 4) {};
		\node [style=none] (8) at (4, 3) {};
		\node [style=none] (9) at (4, 0) {};
		\node [style=X phase dot] (10) at (3, 0) {$\beta$};
		\node [style=X dot] (11) at (-1, 2) {};
		\node [style=Z dot] (12) at (0, 1) {};
		\node [style=White node] (13) at (8, 2) {};
		\node [style=White node] (14) at (8, 2) {};
		\node [style=Gray node] (15) at (9, 4) {};
		\node [style=Gray node] (16) at (9, 0) {};
		\node [style=Gray node] (17) at (13, 2) {};
		\node [style=Gray node] (18) at (11, 3) {};
		\node [style=Gray node] (19) at (12, 3) {};
		\node [style=Gray node] (20) at (13, 4) {};
		\node [style=White node] (21) at (14, 3) {};
		\node [style=White node] (22) at (14, 3) {};
		\node [style=Gray node] (23) at (13, 0) {};
		\node [style=Gray node] (24) at (10, 2) {};
		\node [style=Gray node] (25) at (11, 1) {};
		\node [style=none] (40) at (-6, -3) {};
		\node [style=none] (41) at (-6, -7) {};
		\node [style=Z phase dot] (42) at (-5, -3) {$\alpha$};
		\node [style=Z dot] (43) at (-4, -7) {};
		\node [style=X dot] (44) at (2, -5) {};
		\node [style=Z dot] (45) at (0, -4) {};
		\node [style=Z dot] (46) at (3, -4) {};
		\node [style=X dot] (47) at (1, -3) {};
		\node [style=none] (48) at (4, -4) {};
		\node [style=none] (49) at (4, -7) {};
		\node [style=X phase dot] (50) at (-2, -7) {$\beta$};
		\node [style=X dot] (51) at (-3, -5) {};
		\node [style=Z dot] (52) at (-1, -6) {};
		\node [style=Text node] (53) at (6, -4.375) {{\tiny Optimized\\[-3pt]ZX-diagram}};
		\node [style=White node] (54) at (8, -5) {};
		\node [style=White node] (55) at (8, -5) {};
		\node [style=Gray node] (56) at (9, -5) {};
		\node [style=Gray node] (57) at (10, -5) {};
		\node [style=Gray node] (58) at (16, -5) {};
		\node [style=Gray node] (59) at (14, -5) {};
		\node [style=Gray node] (60) at (17, -5) {};
		\node [style=Gray node] (61) at (15, -5) {};
		\node [style=White node] (62) at (18, -5) {};
		\node [style=White node] (63) at (18, -5) {};
		\node [style=Gray node] (64) at (12, -5) {};
		\node [style=Gray node] (65) at (11, -5) {};
		\node [style=Gray node] (66) at (13, -5) {};
		\node [style=none] (67) at (8.5, -3.25) {};
		\node [style=none] (68) at (8.5, -6.75) {};
		\node [style=none] (69) at (9.5, -3.25) {};
		\node [style=none] (70) at (9.5, -6.75) {};
		\node [style=none] (71) at (10.5, -3.25) {};
		\node [style=none] (72) at (10.5, -6.75) {};
		\node [style=none] (73) at (11.5, -3.25) {};
		\node [style=none] (74) at (11.5, -6.75) {};
		\node [style=none] (75) at (12.5, -3.25) {};
		\node [style=none] (76) at (12.5, -6.75) {};
		\node [style=none] (77) at (13.5, -3.25) {};
		\node [style=none] (78) at (13.5, -6.75) {};
		\node [style=none] (79) at (14.5, -3.25) {};
		\node [style=none] (80) at (14.5, -6.75) {};
		\node [style=none] (81) at (15.5, -3.25) {};
		\node [style=none] (82) at (15.5, -6.75) {};
		\node [style=none] (83) at (16.5, -3.25) {};
		\node [style=none] (84) at (16.5, -6.75) {};
		\node [style=none] (85) at (17.5, -3.25) {};
		\node [style=none] (86) at (17.5, -6.75) {};
		\node [style=Label node] (87) at (8.5, -7) {$2$};
		\node [style=Label node] (88) at (9.5, -7) {$3$};
		\node [style=Label node] (89) at (10.5, -7) {$4$};
		\node [style=Label node] (90) at (11.5, -7) {$4$};
		\node [style=Label node] (91) at (12.5, -7) {$5$};
		\node [style=Label node] (92) at (13.5, -7) {$4$};
		\node [style=Label node] (93) at (14.5, -7) {$5$};
		\node [style=Label node] (94) at (15.5, -7) {$4$};
		\node [style=Label node] (95) at (16.5, -7) {$3$};
		\node [style=Label node] (96) at (17.5, -7) {$2$};
		\node [style=none] (105) at (7, -5) {};
		\node [style=none] (106) at (5, -5) {};
		\node [style=none] (107) at (7, 2) {};
		\node [style=none] (108) at (5, 2) {};
		\node [style=Text node] (109) at (6, 2.625) {\tiny{Signature\\[-3pt]graph}};
		\node [style=none] (110) at (11.25, -2.75) {};
		\node [style=none] (111) at (11.25, -0.75) {};
		\node [style=Text node] (112) at (13, -1.75) {\tiny{Fixed-endvertices\\[-3pt]cutwidth problem}};
	\end{pgfonlayer}
	\begin{pgfonlayer}{edgelayer}
		\draw (5) to (4);
		\draw (5) to (7);
		\draw (7) to (6);
		\draw (4) to (6);
		\draw (2) to (7);
		\draw (1.center) to (3);
		\draw (0.center) to (2);
		\draw (6) to (8.center);
		\draw (3) to (10);
		\draw (11) to (2);
		\draw (10) to (9.center);
		\draw (3) to (11);
		\draw (12) to (11);
		\draw (12) to (4);
		\draw (12) to (10);
		\draw (11) to (5);
		\draw (18) to (17);
		\draw (18) to (20);
		\draw (20) to (19);
		\draw (17) to (19);
		\draw (15) to (20);
		\draw (14) to (16);
		\draw (13) to (15);
		\draw (19) to (21);
		\draw (16) to (23);
		\draw (24) to (15);
		\draw (23) to (22);
		\draw (16) to (24);
		\draw (25) to (24);
		\draw (25) to (17);
		\draw (25) to (23);
		\draw (24) to (18);
		\draw (45) to (44);
		\draw (45) to (47);
		\draw (47) to (46);
		\draw (44) to (46);
		\draw (42) to (47);
		\draw (41.center) to (43);
		\draw (40.center) to (42);
		\draw (46) to (48.center);
		\draw (43) to (50);
		\draw (51) to (42);
		\draw (50) to (49.center);
		\draw (43) to (51);
		\draw (52) to (51);
		\draw (52) to (44);
		\draw (52) to (50);
		\draw (51) to (45);
		\draw [in=120, out=60, looseness=1.50] (59) to (58);
		\draw (59) to (61);
		\draw [in=240, out=-60, looseness=1.50] (61) to (60);
		\draw (58) to (60);
		\draw [in=120, out=60, looseness=0.75] (56) to (61);
		\draw [in=120, out=60, looseness=1.50] (55) to (57);
		\draw (54) to (56);
		\draw (60) to (62);
		\draw [in=120, out=60, looseness=1.50] (57) to (64);
		\draw [in=-60, out=240, looseness=1.50] (65) to (56);
		\draw [in=240, out=-60, looseness=0.75] (64) to (63);
		\draw (57) to (65);
		\draw [in=660, out=-120, looseness=1.50] (66) to (65);
		\draw [in=-105, out=-75] (66) to (58);
		\draw (66) to (64);
		\draw [in=105, out=75] (65) to (59);
		\draw [style=Dashed grey] (67.center) to (68.center);
		\draw [style=Dashed grey] (69.center) to (70.center);
		\draw [style=Dashed grey] (71.center) to (72.center);
		\draw [style=Dashed grey] (73.center) to (74.center);
		\draw [style=Dashed grey] (75.center) to (76.center);
		\draw [style=Dashed grey] (77.center) to (78.center);
		\draw [style=Dashed grey] (79.center) to (80.center);
		\draw [style=Dashed grey] (81.center) to (82.center);
		\draw [style=Dashed grey] (83.center) to (84.center);
		\draw [style=Dashed grey] (85.center) to (86.center);
		\draw [style=Arrow] (105.center) to (106.center);
		\draw [style=Arrow] (108.center) to (107.center);
		\draw [style=Arrow] (111.center) to (110.center);
	\end{pgfonlayer}
\end{tikzpicture}
    \caption{Example demonstrating the procedure used for finding an optimal arrangement of the spiders within a given ZX-diagram to minimize the number of logical qubits required for its implementation via lattice surgery operations.
    We first compute the signature graph $\mathcal{G}_\mathcal{D}$ associated with the ZX-diagram.
    Then, we solve the fixed-endvertices cutwidth problem for $\mathcal{G}_\mathcal{D}$.
    Finally, we rearrange the order of the spiders in the ZX-diagram based on the ordering obtained.
    We obtain an equivalent ZX-diagram which can be implemented using $5$ logical qubits instead of $9$ for the initial ZX-diagram.
}
    \label{fig:cutwidth_example}
\end{figure}

We now prove the following theorem, which states that the problem of minimizing the maximum number of wires in any vertical cut of a given ZX-diagram by moving its spiders and bending (or unbending) its wires is equivalent to the fixed-endvertices cutwidth problem.

\begin{theorem}\label{thm:cutwidth_equivalence}
    Let $(\mathcal{G}_\mathcal{D}, \mathcal{I}, \mathcal{O})$ be the signature of a given ZX-diagram $\mathcal{D}$.
    And let $\mathcal{G}_\mathcal{D}'$ be constructed from $\mathcal{G}_\mathcal{D}$ by merging its subset of vertices $\mathcal{I}$ into a single vertex $u$ and its subset of vertices $\mathcal{O}$ into a single vertex $w$.
    Then, solving the fixed-endvertices problem for $\mathcal{G}_\mathcal{D}'$ with endvertices $u$ and $w$ is equivalent to solving the problem of finding a ZX-diagram $\mathcal{D}'$ which doesn't have any vertical wire, which can be transformed into $\mathcal{D}$ only by moving its spiders and bending (or unbending) its wires, and such that the maximum number of wires in any vertical cut of $\mathcal{D}'$ is minimized.
\end{theorem}

\begin{proof}
    We prove the equivalence between these two problems by demonstrating that an optimal solution to either problem provides an equivalent solution to the other problem.

    Let $f_{u, w}$ be a solution to the fixed-endvertices cutwidth problem for $\mathcal{G}_{\mathcal{D}}'$ where $u$ and $w$ are the fixed first and last vertices of the ordering $f_{u, w}$.
    And let $\mathcal{D}'$ be a ZX-diagram constructed from $\mathcal{D}$ by reordering its spiders according to the ordering $f_{u, w}$ and by straightening its wires.
    Since the spiders of $\mathcal{D}'$ are ordered in the same way as the vertices in $\mathcal{G}_{\mathcal{D}}'$ with respect to $f_{u, w}$, it follows that the maximum number of wires in any vertical cut of $\mathcal{D}'$ is equal to $cw_{f_{u,w}}(\mathcal{G}_{\mathcal{D}}')$.

    Let $\mathcal{D}'$ be a ZX-diagram which doesn't have any vertical wire and which can be transformed into $\mathcal{D}$ only by moving its spiders and bending (or unbending) its wires, and such that the maximum number of wires in any vertical cut of $\mathcal{D}'$, noted $\kappa$, is minimized.
    We first show that we can modify $\mathcal{D}'$ to obtain a strict total order of its spiders without increasing the maximum number of wires in any vertical cut of the diagram.
    Let $U = \{u_1, \ldots, u_n\}$ be a set of spiders in $\mathcal{D}'$ that are vertically aligned.
    And let $\alpha \leq \kappa$ be the number of wires in the vertical cut of $\mathcal{D}'$ positioned just before the spiders in $U$, and $\beta \leq \kappa$ be the number of wires in the vertical cut of $\mathcal{D}'$ positioned just after the spiders in $U$.
    Let $u_i \in U$ be the spider such that its number of input or output wires is both lower or equal to the number of input wires of every other spider in $U$ and the number of output wires of every other spider in $U$.
    If the number of output wires of $u_i$ is lower than its number of input wires, then we modify the diagram $\mathcal{D}'$ by placing the spider $u_i$ just before the set of spiders $U \setminus \{u_i\}$.
    Otherwise, we modify the diagram $\mathcal{D}'$ by placing the spider $u_i$ just after the set of spiders $U \setminus \{u_i\}$.
    Let $\gamma$ be the number of wires in the vertical cut positioned between the spider $u_i$ and the spiders $U \setminus \{u_i\}$.
    If $u_i$ is placed before the set of spiders $U \setminus \{u_i\}$, then we have $\gamma \leq \alpha$, which implies that $\gamma \leq \kappa$.
    Otherwise, if $u_i$ is placed after the set of spiders $U \setminus \{u_i\}$, then we have $\gamma \leq \beta$, which also implies that $\gamma \leq \kappa$.
    The number of wires in all the other vertical cuts of the diagram remains the same.
    Then, we can repeat this process for every set of vertically aligned spiders in $\mathcal{D}'$ to obtain an equivalent ZX-diagram $\tilde{\mathcal{D}}'$ in which no spiders are vertically aligned, and such that the maximum number of wires in any vertical cut of the diagram is equal to $\kappa$.
    Moreover, note that the ZX-diagram obtained by straightening the wires of $\tilde{\mathcal{D}}'$ also has a maximum number of wires in any vertical cut of the diagram equal to $\kappa$.
    Let $f'_{u,w}$ be the ordering of the vertices in $\mathcal{G}_{\mathcal{D}}'$ corresponding to the ordering of the spiders in $\tilde{\mathcal{D}}'$, then we have $cw_{f'_{u,w}}(\mathcal{G}_{\mathcal{D}}') = \kappa$.
\end{proof}

The close relation between the cutwidth problem and the fixed-endvertices cutwidth problem can be leveraged to prove the NP-hardness of the fixed-endvertices cutwidth problem as follows.

\begin{theorem}\label{thm:cutwidth_np_hard}
    The fixed-endvertices cutwidth problem is NP-hard.
\end{theorem}

\begin{proof}
    Let $\mathcal{G}=(V, E)$ be a graph with vertex set $V$ and edge set $E$.
    And let $f_{u, w}$ be an optimal solution to the fixed-endvertices for $\mathcal{G}$ where $u$ and $w$ are the fixed first and last vertices of the ordering $f_{u, w}$.
    Then, by definition, the cutwidth of $\mathcal{G}$ satisfies
    \begin{equation}\label{eq:cutwidth_red}
        cw(\mathcal{G}) = \min_{\substack{u, w \in V\\u\neq w}} cw_{f_{u, w}}(\mathcal{G}).
    \end{equation}
    Therefore, we can solve the cutwidth problem by solving the fixed-endvertices cutwidth problem a polynomial number of times, and then use these solutions to compute the cutwidth of $\mathcal{G}$ via Equation~\ref{eq:cutwidth_red} in polynomial time.
    Thus, the fixed-endvertices cutwidth problem is at least as hard as the cutwidth problem, which is an NP-hard problem~\cite{gavril1977some}.
\end{proof}

From Theorem~\ref{thm:cutwidth_equivalence} and Theorem~\ref{thm:cutwidth_np_hard}, it follows that the problem of minimizing the number of logical qubits required to implement a given ZX-diagram using lattice surgery operations by moving its spiders and bending (or unbending) its wires is an NP-hard problem.
In order to solve this problem more efficiently, it can be useful to first apply some transformation of the initial ZX-diagram to minimize its number of spiders, for instance by using the spider fusion rule presented in Section~\ref{sub:zx_calculus}.
Also, spiders having only two incident wires don't have any impact on the optimization problem.
Each one of these spiders can be replaced by a single wire and inserted back on their associated wire in the optimized ZX-diagram without altering the obtained solution.

\subsection{Qubit-count optimization using the fusion rule}\label{sub:qubits_opt_fusion_rule}

The spider fusion rule, presented in Section~\ref{sub:zx_calculus}, can enable significant reductions in the number of qubits required to implement a given ZX-diagram through lattice surgery operations.
The most basic example is for implementing a multi-target CNOT gate in lattice surgery, without connectivity constraints.
It can either be done with a constant depth but at a linear cost in the number of logical qubits, or, at the extreme opposite, it can be done with only one intermediate logical qubit but at the expense of a linear depth with respect to the number of target qubits:
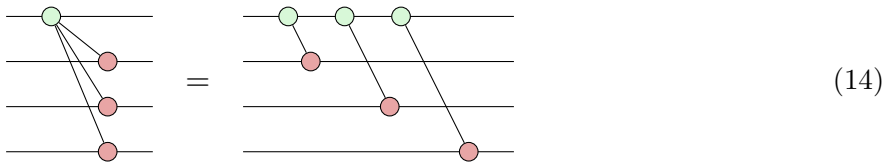
\begin{equation}
    \begin{tikzpicture}
	\begin{pgfonlayer}{nodelayer}
		\node [style=none] (0) at (0, 2) {};
		\node [style=none] (1) at (0, 1) {};
		\node [style=none] (2) at (0, 0) {};
		\node [style=none] (3) at (0, -1) {};
		\node [style=none] (4) at (3.25, 2) {};
		\node [style=none] (5) at (3.25, 1) {};
		\node [style=none] (6) at (3.25, 0) {};
		\node [style=none] (7) at (3.25, -1) {};
		\node [style=Z dot] (8) at (1, 2) {};
		\node [style=X dot] (9) at (2.25, 1) {};
		\node [style=X dot] (10) at (2.25, 0) {};
		\node [style=X dot] (11) at (2.25, -1) {};
		\node [style=none] (12) at (5.25, 2) {};
		\node [style=none] (13) at (5.25, 1) {};
		\node [style=none] (14) at (5.25, 0) {};
		\node [style=none] (15) at (5.25, -1) {};
		\node [style=none] (16) at (11.25, 2) {};
		\node [style=none] (17) at (11.25, 1) {};
		\node [style=none] (18) at (11.25, 0) {};
		\node [style=none] (19) at (11.25, -1) {};
		\node [style=Z dot] (20) at (6.25, 2) {};
		\node [style=X dot] (21) at (6.75, 1) {};
		\node [style=X dot] (22) at (8.5, 0) {};
		\node [style=X dot] (23) at (10.25, -1) {};
		\node [style=none] (24) at (4.25, 0.5) {$=$};
		\node [style=Z dot] (25) at (7.5, 2) {};
		\node [style=Z dot] (26) at (8.75, 2) {};
	\end{pgfonlayer}
	\begin{pgfonlayer}{edgelayer}
		\draw (0.center) to (4.center);
		\draw (1.center) to (5.center);
		\draw (2.center) to (6.center);
		\draw (3.center) to (7.center);
		\draw (11) to (8);
		\draw (8) to (10);
		\draw (8) to (9);
		\draw [in=180, out=0] (12.center) to (16.center);
		\draw (13.center) to (17.center);
		\draw (14.center) to (18.center);
		\draw (15.center) to (19.center);
		\draw (20) to (21);
		\draw (25) to (22);
		\draw (26) to (23);
	\end{pgfonlayer}
\end{tikzpicture}
\end{equation}

A simple approach for exploiting this kind of optimizations could be to use the spider fusion rule to replace each spider of the diagram by a chain of spiders, each incident to three wires, and then use the results of Section~\ref{sub:qubits_opt_vertex_ordering} to optimize the number of qubits.
However, the number of ways to unfuse a spider in this manner grows exponentially with respect to the number of wires connected to it, which makes this approach impractical.
In this section, we show how the spider fusion rule can be better incorporated into our qubit-count optimization problem and how it relates to the problem of finding the pathwidth of a graph.
We first define the pathwidth as follows.

\begin{definition}[Pathwidth]
    Let $\mathcal{G}=(V,E)$ be a graph with vertex set $V$ and edge set $E$.
    The pathwidth of a vertex $v \in V$ with respect to a vertex ordering $f$, denoted by $pw_f(v)$, is defined as:
    \begin{equation}
        pw_f(v) = \lvert \{u \mid f(u) \leq f(v) < f(w), (u, w) \in E \} \rvert.
    \end{equation}
    The pathwidth of $\mathcal{G}$ with respect to $f$, denoted by $pw_f(\mathcal{G})$, is defined as: 
    \begin{equation}
        pw_f(\mathcal{G}) = \max_{v\in V} pw_f(v).
    \end{equation}
    And the pathwidth of $\mathcal{G}$, denoted by $pw(\mathcal{G})$, is defined as the minimum $pw_f(\mathcal{G})$ value over all possible orderings:
    \begin{equation}
        pw(\mathcal{G}) = \min_{f \in S_V} pw_f(\mathcal{G}) = \min_f \max_{v\in V} pw_f(v)
    \end{equation}
    where $S_V$ is the set of all orderings of the vertices in $V$.
\end{definition}

This definition of pathwidth was initially referred to as the vertex separation number, before its equivalence to the pathwidth was shown~\cite{kinnersley1992vertex}.
The pathwidth problem consists in finding the pathwidth $pw(\mathcal{G})$ of a given graph $\mathcal{G}$.
This problem is known to be NP-hard~\cite{kashiwabara1979np, ohtsuki1979one}.

We now show how rearranging the order of the spiders of a ZX-diagram and using the spider fusion rule to minimize the number of qubits required to implement it through surgery operations relates to the pathwidth problem.
First, we can notice that the maximum number of wires in any vertical cut of a given ZX-diagram can be optimized so that it is always equal to $n+1$, where $n$ is the number of spiders in the initial ZX-diagram.
To do so, we assign an horizontal line to each spider of the ZX-diagram, resulting in $n$ distinct horizontal lines.
Then, for each spider $u$ connected to spiders $v_1, \ldots, v_k$, we can use the fusion rule to decompose $u$ it into a sequence of spiders $u_1, \dots, u_k$ where $u_i$ is connected to $v_i$ and $u_i$ lies on the horizontal line associated with $u$.
This can be done such that the resulting ZX-diagram has a circuit-like form, where the horizontal wires are corresponding to qubits and non-horizontal wires are corresponding to gates acting on these qubits and are performed sequentially.
As a result, each vertical cut of the transformed ZX-diagram intersects $n$ horizontal wires and at most one non-horizontal wire.
An illustrative example is provided in Figure~\ref{fig:trivial_path_decomposition_example}.

\begin{figure}[t]
    \centering
    \begin{tikzpicture}
	\begin{pgfonlayer}{nodelayer}
		\node [style=none] (13) at (5.5, 0.75) {};
		\node [style=none] (14) at (5.5, -1.75) {};
		\node [style=none] (15) at (9.5, 0.75) {};
		\node [style=none] (16) at (9.5, -1.75) {};
		\node [style=Z dot] (17) at (6.5, 0.75) {};
		\node [style=X dot] (18) at (8.5, 0.25) {};
		\node [style=X dot] (19) at (8.5, -1.75) {};
		\node [style=Z dot] (20) at (6.5, -0.25) {};
		\node [style=X dot] (21) at (8.5, -0.75) {};
		\node [style=Z dot] (22) at (6.5, -1.25) {};
		\node [style=none] (23) at (11.5, 2) {};
		\node [style=none] (24) at (11.5, -3) {};
		\node [style=none] (25) at (26, 2) {};
		\node [style=none] (26) at (26, -3) {};
		\node [style=Z dot] (27) at (15.5, 2) {};
		\node [style=Z dot] (32) at (16.5, -2) {};
		\node [style=X dot] (33) at (25.75, 1) {};
		\node [style=X dot] (34) at (11.75, 1) {};
		\node [style=X dot] (35) at (11.75, -1) {};
		\node [style=X dot] (36) at (25.75, -1) {};
		\node [style=Z dot] (37) at (25.75, 0) {};
		\node [style=Z dot] (38) at (11.75, 0) {};
		\node [style=Z dot] (39) at (11.75, -2) {};
		\node [style=Z dot] (40) at (25.75, -2) {};
		\node [style=Z dot] (41) at (13, 2) {};
		\node [style=X dot] (42) at (13.25, 1) {};
		\node [style=X dot] (43) at (15, -1) {};
		\node [style=Z dot] (44) at (14.25, 2) {};
		\node [style=X dot] (45) at (18, 1) {};
		\node [style=Z dot] (46) at (18.25, 0) {};
		\node [style=Z dot] (47) at (20, -2) {};
		\node [style=X dot] (48) at (19.25, 1) {};
		\node [style=X dot] (49) at (20.75, -1) {};
		\node [style=X dot] (50) at (22.5, -3) {};
		\node [style=Z dot] (51) at (20.5, 0) {};
		\node [style=Z dot] (52) at (21.75, 0) {};
		\node [style=X dot] (53) at (18, -3) {};
		\node [style=Z dot] (54) at (16.75, 2) {};
		\node [style=X dot] (55) at (23, -1) {};
		\node [style=Z dot] (56) at (23.25, -2) {};
		\node [style=Z dot] (57) at (24.5, -2) {};
		\node [style=X dot] (58) at (24.75, -3) {};
		\node [style=none] (59) at (10.5, -0.5) {$=$};
	\end{pgfonlayer}
	\begin{pgfonlayer}{edgelayer}
		\draw (13.center) to (15.center);
		\draw (14.center) to (16.center);
		\draw (19) to (17);
		\draw (17) to (18);
		\draw (20) to (18);
		\draw (20) to (19);
		\draw (22) to (18);
		\draw (22) to (21);
		\draw (22) to (19);
		\draw (20) to (21);
		\draw (17) to (21);
		\draw [style=Hadamard edge, in=120, out=-120] (17) to (22);
		\draw (23.center) to (25.center);
		\draw (24.center) to (26.center);
		\draw [style=Hadamard edge] (27) to (32);
		\draw (34) to (33);
		\draw (37) to (38);
		\draw (35) to (36);
		\draw (40) to (39);
		\draw (41) to (42);
		\draw (44) to (43);
		\draw (45) to (46);
		\draw (47) to (48);
		\draw (51) to (49);
		\draw (52) to (50);
		\draw (53) to (54);
		\draw (55) to (56);
		\draw (58) to (57);
	\end{pgfonlayer}
\end{tikzpicture}
    \caption{Example demonstrating how the maximum number of wires in any vertical cut of a given ZX-diagram can be optimized so that it is equal to $n+1$, where $n$ is the number of spiders in the initial ZX-diagram.
    One horizontal line is associated with each spider of the initial ZX-diagram.
    Then, the fusion rule is applied to ensure that any vertical cut of the ZX-diagram is intersecting at most a single wire connecting two spiders in the initial ZX-diagram.}
    \label{fig:trivial_path_decomposition_example}
\end{figure}
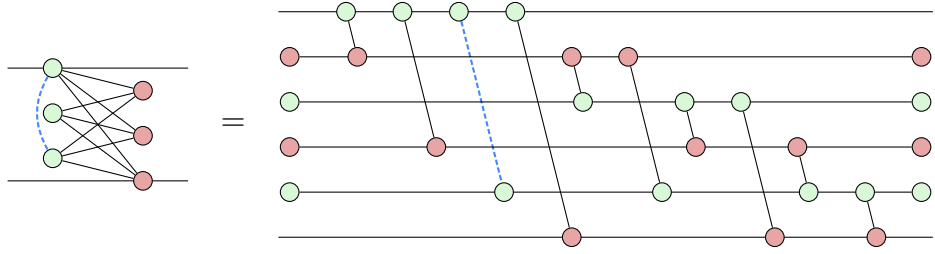

To further optimize the ZX-diagram, instead of associating each horizontal line with a single spider, we can place several spiders on the same horizontal line if certain conditions are met.
We will see that these conditions are closely related to the graph-theoretical notion of path-decomposition, defined as follows~\cite{robertson1983graph}.

\begin{definition}[Path-decomposition]
    Let $\mathcal{G}=(V,E)$ be a graph with vertex set $V$ and edge set $E$.
    A path-decomposition $\mathcal{G}$ is a sequence $X_1, \ldots, X_\ell$ of subsets of $V$ such that:
    \begin{itemize}
        \item for each edge $(u, v) \in E$, there exists $i$ such that $u \in X_i$ and $v \in X_i$;
        \item for all $i,j,k$ such that $1 \leq i \leq j \leq k \leq \ell$, $X_i \cap X_k \subseteq X_j$.
    \end{itemize}
\end{definition}

Let $(\mathcal{G}_\mathcal{D}, \mathcal{I}, \mathcal{O})$ be the signature of a ZX-diagram $\mathcal{D}$, and let $X_1, \ldots, X_\ell$ be a path-decomposition of $\mathcal{G}_\mathcal{D}$ such that $u \in X_1$ for all $u \in \mathcal{I}$ and $u \in X_\ell$ for all $u \in \mathcal{O}$.
Then we can construct a ZX-diagram $\mathcal{D}'$, equivalent to $\mathcal{D}$, from the path-decomposition $X_1, \ldots, X_\ell$ which contains $\ell$ layers of spiders as follows.
First, for every spider $u$ in $\mathcal{D}$ we unfuse it in $\mathcal{D}'$ for every layer $i$ where $u \in X_i$.
Then, for every wire connecting two spiders $u$ and $v$ in $\mathcal{D}$ we connect these two spiders in a layer $i$ of $\mathcal{D}'$ where $u \in X_i$ and $v \in X_i$.
Note that, by definition, the existence of a $X_i$ in the the path-decomposition is guaranteed.
Then, similarly to Figure~\ref{fig:trivial_path_decomposition_example}, we obtain a ZX-diagram which has a circuit-like form, where the wires resulting from the application of the spider fusion rule are horizontal and are corresponding to qubits, and all the other non-horizontal wires are corresponding to gates that can be realized sequentially.
In fact, the example of Figure~\ref{fig:trivial_path_decomposition_example} corresponds to the trivial case where $\ell = 1$ and all spiders are in $X_1$.
Any vertical cut of $\mathcal{D}'$ intersects at most $\max_i{\lvert X_i \rvert}$ horizontal wires and at most $1$ non-horizontal wire.
Therefore, $\mathcal{D}'$ can be implemented trough lattice surgery operations by using at most $\max_i{\lvert X_i \rvert} + 1$ logical qubits.
Thus, we want to find a path-decomposition $X_1, \ldots, X_\ell$ of $\mathcal{G}_\mathcal{D}$ such that $\max_i{\lvert X_i \rvert}$ is minimized, and where $u \in X_1$ for all $u \in \mathcal{I}$ and $u \in X_\ell$ for all $u \in \mathcal{O}$.
The problem of finding a path-decomposition $X_1, \ldots, X_\ell$ of a graph $\mathcal{G}$ that minimizes $\max_i{\lvert X_i \rvert}$ is equivalent to the problem of finding an ordering $f$ of its vertices such that $pw_f(\mathcal{G})$ is minimized~\cite{kinnersley1992vertex}.
Therefore, our problem corresponds to the problem of finding an ordering $f$ of the vertices of $\mathcal{G}_\mathcal{D}$ with minimum pathwidth $pw_f(\mathcal{G}_\mathcal{D})$, but where the first and last set of vertices of the ordering $f$ are fixed such that $f(u) < f(v) < f(w)$ for all $u \in \mathcal{I}$, $v \in V \setminus \{\mathcal{I}, \mathcal{O}\}$, $w \in \mathcal{O}$, where $V$ is the vertex set of $\mathcal{G}_\mathcal{D}$.
We refer to this specific problem as the fixed-endbags pathwidth problem.

\begin{problem}[Fixed-endbags pathwidth]
    Let $\mathcal{G}=(V,E)$ be a graph with vertex set $V$ and edge set $E$ and let $U, W \subseteq V$.
    Find an ordering $f$ of the vertices $V$ satisfying $f(u) < f(v) < f(w)$ for all $u \in U$, $v \in V \setminus \{U, W\}$, $w \in W$ and such that $pw_f(\mathcal{G})$ is minimized.
\end{problem}

Let $(\mathcal{G}_\mathcal{D}, \mathcal{I}, \mathcal{O})$ be the signature of a ZX-diagram $\mathcal{D}$.
We will refer to the pathwidth of $\mathcal{D}$, noted $pw(\mathcal{D})$, as the value $pw_f(\mathcal{G}_\mathcal{D})$ such that $f$ is solving the fixed-endbags pathwidth problem where $\mathcal{I}$ and $\mathcal{O}$ are the fixed first and last bags respectively.
We now prove the following theorem, which states that a ZX-diagram $\mathcal{D}$ cannot be implemented through lattice surgery operations with less than $pw(\mathcal{D})$ logical qubits if $\mathcal{D}$ can only be transformed by using the spider fusion rule, by moving its spiders and by bending its wires.

\begin{theorem}\label{thm:pathwidth_lower_bound}
    Let $\mathcal{D}$ be a ZX-diagram in which no spider of the same color are connected by a wire.
    And let $\mathcal{D}'$ be a ZX-diagram which can be transformed into $\mathcal{D}$ by using the spider fusion rule, by moving its spiders and by bending its wires.
    Then, the maximum number of wire in any vertical cut of $\mathcal{D}'$ is at least $pw(\mathcal{D})$, i.e.\ at least $pw(\mathcal{D})$ logical qubits are required to implement $\mathcal{D}'$ through lattice surgery operations.
\end{theorem}

\begin{proof}
    Let $\mathcal{D}'$ be a ZX-diagram which can be transformed into $\mathcal{D}$ by using the spider fusion rule, by moving its spiders and by bending its wires; and let $\kappa$ be the maximum number of wires in any vertical cut of $\mathcal{D}'$.
    As demonstrated in the proof of Theorem~\ref{thm:cutwidth_equivalence}, $\mathcal{D}'$ can be transformed into an equivalent ZX-diagram $\tilde{\mathcal{D}}'$ such that the spiders of $\tilde{\mathcal{D}}'$ have a strict total order (i.e.\ no two spiders in $\tilde{\mathcal{D}}'$ are vertically aligned) and such that the maximum number of wires in any vertical cut of $\tilde{\mathcal{D}}'$ is also equal to $\kappa$.

    Let $\mathcal{G}_\mathcal{D}$ be the signature graph of $\mathcal{D}$, $\mathcal{G}_{\tilde{\mathcal{D}}'}$ be the signature graph of $\tilde{\mathcal{D}}'$, and let $f$ be the ordering of the vertices in $\mathcal{G}_{\tilde{\mathcal{D}}'}$ corresponding to the ordering of the spiders in $\tilde{\mathcal{D}}'$.
    Then we have $cw_f(\mathcal{G}_{\tilde{\mathcal{D}}'}) = \kappa$.
    The cutwidth is always lower or equal to the pathwidth, therefore we have $pw(\tilde{\mathcal{D}}') \leq pw_f(\mathcal{G}_{\tilde{\mathcal{D}}'}) \leq cw_f(\mathcal{G}_{\tilde{\mathcal{D}}'}) = \kappa$.
    Moreover, the graph $\mathcal{G}_\mathcal{D}$ can be obtain from $\mathcal{G}_{\tilde{\mathcal{D}}'}$ by performing a sequence of edge contraction operations.
    Performing an edge contraction on a graph does not increase its pathwidth.
    Therefore, we have $pw(\mathcal{D}) = pw(\mathcal{G}_\mathcal{D}) \leq pw(\mathcal{G}_{\tilde{\mathcal{D}}'}) = pw(\tilde{\mathcal{D}}')$.
    It follows that $pw(\mathcal{D}) \leq pw(\tilde{\mathcal{D}}') \leq pw_f(\mathcal{G}_{\tilde{\mathcal{D}}'}) \leq cw_f(\mathcal{G}_{\tilde{\mathcal{D}}'}) = \kappa$.
\end{proof}

Consider the algorithm whose pseudo-code is given in Algorithm~\ref{alg:zx_fixed_endbags_pathwidth} and which takes a ZX-diagram $\mathcal{D}$ as input.
The algorithm starts by transforming $\mathcal{D}$ to minimize its number of spiders using the fusion rule.
This transformation ensures that the resulting ZX-diagram doesn't contain any pair of Z spiders or X spiders connected by a wire.
Then, the algorithm computes a solution $f$ to the fixed-endbags pathwidth problem for $\mathcal{G}_\mathcal{D}$ with endbags $\mathcal{I}$ and $\mathcal{O}$, where $(\mathcal{G}_\mathcal{D}, \mathcal{I}, \mathcal{O})$ is the signature of $\mathcal{D}$.
The ordering $f$ is then used to construct an optimized ZX-diagram $\mathcal{D}'$ equivalent to $\mathcal{D}$.
To do so, we can first notice that each vertex $u$ in $\mathcal{G}_\mathcal{D}$ can be associated with an interval that spans from $f(u)$ to the maximum value $f(v)$ where $v$ is a neighbor of $u$ in $\mathcal{G}_\mathcal{D}$ or $v = u$.
By constructing the intersection graph of these intervals, we obtain an interval supergraph of $\mathcal{G}$ in which the clique number minus one is equal to $pw_f(\mathcal{G}_\mathcal{D})$~\cite{bodlaender1998partial}.
Each interval represents a space in which the associated spider in $\mathcal{D}$ can be unfused.
Then, we obtain a ZX-diagram which has a circuit-like form where the wires resulting from the application of the spider fusion rule are horizontal and are corresponding to qubits, and all the other non-horizontal wires are corresponding to gates that can be realized sequentially.
An illustrative example of the complete procedure is provided in Figure~\ref{fig:pathwidth_example}.

\begin{algorithm}[t]
    \caption{Optimize a ZX-diagram using the fixed-endbags pathwidth problem}
    \label{alg:zx_fixed_endbags_pathwidth}
	\SetAlgoLined
	\SetArgSty{textnormal}
	\SetKwFunction{proc}{}
	\SetKwInput{KwInput}{Input}
	\SetKwInput{KwOutput}{Output}
    \KwInput{A ZX-diagram $\mathcal{D}$.}
    \KwOutput{An optimized ZX-diagram $\mathcal{D}'$ equivalent to $\mathcal{D}$.}
	\SetKwProg{Fn}{procedure}{}{}
    $\mathcal{D} \leftarrow$ $\mathcal{D}$ where each pair of connected Z spiders and X spiders have been fused \\
    $(\mathcal{G}_\mathcal{D}, \mathcal{I}, \mathcal{O}) \leftarrow$ signature of $\mathcal{D}$ \\
    $f \leftarrow$ solution to the fixed-endbags pathwidth problem for $\mathcal{G}_\mathcal{D}$ with endbags $\mathcal{I}$ and $\mathcal{O}$ \\
    $\mathcal{D}' \leftarrow$ copy of $\mathcal{D}$ \\
    \ForEach{spider $u$ of $\mathcal{D}$ connected to $n > 1$ spiders $v_1, \ldots, v_n$ such that $f(v_i) < f(v_{i+1})$}{
        $\mathcal{D}' \leftarrow$ copy of $\mathcal{D}'$ where $u$ is unfused into $n$ spiders $u_1, \ldots, u_n$ such that $u_i$ is connected to $u_{i+1}$ and $v_{i}$ \\
        $f \leftarrow$ copy of $f$ where $f(u_i) = f(u)$ \\
    }
    $\mathcal{D}' \leftarrow$ copy of $\mathcal{D}'$ where each spider $u$ appears before another spider $w$ if and only if $f(u) < f(w)$ or $f(u) = f(w)$ and $f(u)$ is connected to a spider $v$ such that $f(u) < f(v)$ \\
    \Return $\mathcal{D}'$ with all its wires straightened
\end{algorithm}

\begin{figure}[tp]
    \centering
    \input{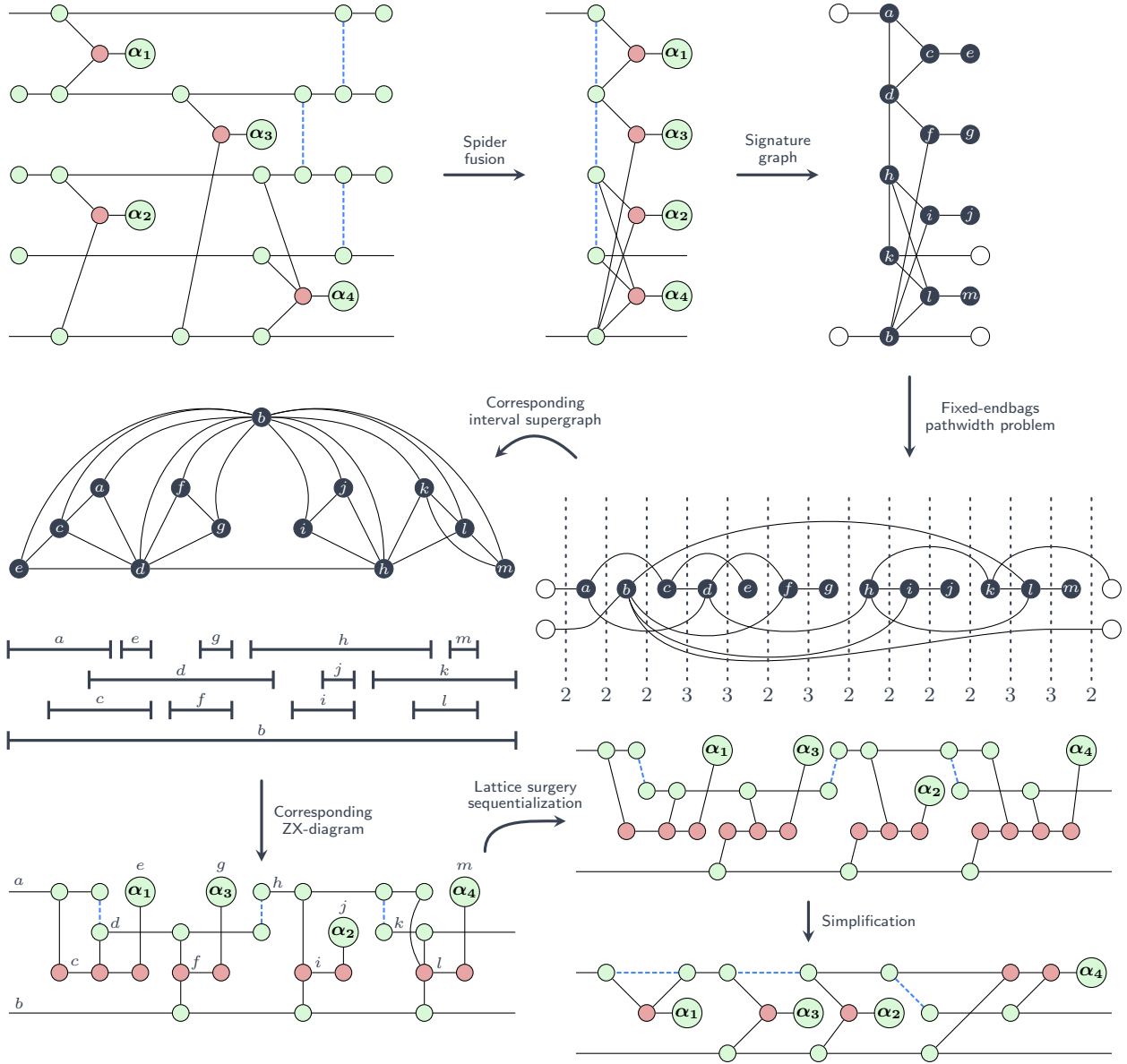}
    \caption{Example demonstrating the procedure used for optimizing the number of qubits required to implement a given ZX-diagram using lattice surgery operations by rearranging its spiders and using the spider fusion rule.
    We first fuse the spiders and compute the signature graph $\mathcal{G}_\mathcal{D}$ associated with the resulting ZX-diagram $\mathcal{D}$.
    Then, we solve the fixed-endbags pathwidth problem for $\mathcal{G}_\mathcal{D}$.
    The solution gives us an interval graph where each interval corresponds to a space in which the associated spider of $\mathcal{D}$ must lie and can be unfused.
    We construct the associated ZX-diagram in which the horizontal wires (except the input and output wires) result from the application of the spider fusion rule, and each vertical wire corresponds to a wire in $\mathcal{D}$.
    The lattice surgery operations associated with the vertical wires can then be done sequentially, such that there is no split or merge operation involving more than 2 logical qubits simultaneously and more than one non-horizontal wire.
    Finally, we can simplify the diagram by applying the identity rule and performing parallel operations without increasing the number of wires in any vertical cut of the diagram.
    We obtain an optimized ZX-diagram (equivalent to $\mathcal{D}$) describing a sequence of lattice surgery operations that can be performed using $4$ logical qubits.}
    \label{fig:pathwidth_example}
\end{figure}

The following theorem provides an upper bound on the maximum number of wires in any vertical cut of the ZX-diagram $\mathcal{D}'$ produced by Algorithm~\ref{alg:zx_fixed_endbags_pathwidth}.

\begin{theorem}\label{thm:pathwidth_upper_bound}
    Let $\mathcal{D}$ be a ZX-diagram with signature graph $\mathcal{G}_{D}$, and let $\mathcal{D}'$ be the ZX-diagram produced by Algorithm~\ref{alg:zx_fixed_endbags_pathwidth}.
    Then any vertical cut of $\mathcal{D}'$ intersects at most $pw(\mathcal{D}) + 2$ wires, i.e.\ $\mathcal{D}'$ can be implemented with at most $pw(\mathcal{D}) + 2$ logical qubits through lattice surgery operations.
\end{theorem}

\begin{proof}
    Let $f$ be a solution to the fixed-endbags pathwidth problem for $\mathcal{G}_\mathcal{D}$ with endbags $\mathcal{I}$ and $\mathcal{O}$, where $(\mathcal{G}_\mathcal{D}, \mathcal{I}, \mathcal{O})$ is the signature of $\mathcal{D}$.
    And let $S$ be the set of intervals defined as follows:
    \begin{equation}
        S = \{[f(u), f(v)] \mid u \in V, f(v) = \max_{w \in N[u]} f(w) \}
    \end{equation}
    where $V$ is the vertex set of $\mathcal{G}_\mathcal{D}$ and $N[u]$ denotes the closed neighborhood of the vertex $u$ in $\mathcal{G}_\mathcal{D}$.
    The intervals in $S$ can be arranged on $pw_f(\mathcal{G}_\mathcal{D}) + 1 = pw(\mathcal{D}) + 1$ distinct horizontal lines such that there is no overlapping point between two intervals on any line~\cite{bodlaender1998partial}.
    By construction, all the spiders of $\mathcal{D}'$ are positioned on these intervals.
    Therefore, any vertical cut of $\mathcal{D}'$ intersects at most $pw(\mathcal{D}) + 1$ horizontal wires.
    Moreover, any vertical cut of $\mathcal{D}'$ intersects at most one non-horizontal wire.
    Thus, any vertical cut of $\mathcal{D}'$ intersects at most $pw(\mathcal{D}) + 2$ wires.
\end{proof}

Let $\mathcal{D}$ be a ZX-diagram and let $\mathcal{D}'$ be another ZX-diagram which can be transformed into $\mathcal{D}$ by using the spider fusion rule, by moving its spiders and by bending its wires.
Theorem~\ref{thm:pathwidth_lower_bound} gives us a lower bound of $pw(\mathcal{D})$ and Theorem~\ref{thm:pathwidth_upper_bound} gives us an upper bound of $pw(\mathcal{D}) + 2$ for the number of logical qubits required to implement $\mathcal{D}'$ through lattice surgery operations.
These upper and lower bounds are tight.
In Figure~\ref{fig:pathwidth_opt_examples}, we provide three ZX-diagrams for which the associated optimized ZX-diagram $\mathcal{D}'$ requires at least $pw(\mathcal{D})$, $pw(\mathcal{D}) + 1$, and $pw(\mathcal{D}) + 2$ logical qubits respectively to be implemented through lattice surgery operations.

\begin{figure}[t]
    \centering
    \begin{subfigure}{0.3\textwidth} \centering
        \begin{tikzpicture}
	\begin{pgfonlayer}{nodelayer}
		\node [style=none] (0) at (-3, 3) {};
		\node [style=none] (3) at (-0.5, 3) {};
		\node [style=Z phase dot] (6) at (-1.75, 3) {$\alpha$};
	\end{pgfonlayer}
	\begin{pgfonlayer}{edgelayer}
		\draw (3.center) to (0.center);
	\end{pgfonlayer}
\end{tikzpicture}
        \caption{$pw(\mathcal{D})$}
    \end{subfigure}
    \begin{subfigure}{0.3\textwidth} \centering
        \begin{tikzpicture}
	\begin{pgfonlayer}{nodelayer}
		\node [style=none] (0) at (-3, 3) {};
		\node [style=none] (1) at (-3, 1.75) {};
		\node [style=none] (3) at (0, 3) {};
		\node [style=none] (4) at (0, 1.75) {};
		\node [style=Z dot] (6) at (-2, 3) {};
		\node [style=X dot] (9) at (-1, 1.75) {};
	\end{pgfonlayer}
	\begin{pgfonlayer}{edgelayer}
		\draw (3.center) to (0.center);
		\draw (1.center) to (4.center);
		\draw (6) to (9);
	\end{pgfonlayer}
\end{tikzpicture}
        \caption{$pw(\mathcal{D})+1$}
    \end{subfigure}
    \begin{subfigure}{0.3\textwidth} \centering
        \begin{tikzpicture}
	\begin{pgfonlayer}{nodelayer}
		\node [style=none] (0) at (-3, 3) {};
		\node [style=none] (1) at (-3, 1.75) {};
		\node [style=none] (2) at (-3, 0.5) {};
		\node [style=none] (3) at (0, 3) {};
		\node [style=none] (4) at (0, 1.75) {};
		\node [style=none] (5) at (0, 0.5) {};
		\node [style=Z dot] (6) at (-2, 3) {};
		\node [style=Z dot] (7) at (-2, 1.75) {};
		\node [style=Z dot] (8) at (-2, 0.5) {};
		\node [style=X dot] (9) at (-1, 2.5) {};
	\end{pgfonlayer}
	\begin{pgfonlayer}{edgelayer}
		\draw (3.center) to (0.center);
		\draw (1.center) to (4.center);
		\draw (5.center) to (2.center);
		\draw (6) to (9);
		\draw (9) to (7);
		\draw (8) to (9);
	\end{pgfonlayer}
\end{tikzpicture}
        \caption{$pw(\mathcal{D})+2$}
    \end{subfigure}
    \caption{Example of ZX-diagrams that are requiring $pw(\mathcal{D})$, $pw(\mathcal{D})+1$ and $pw(\mathcal{D})+2$ logical qubits respectively to be implemented through lattice surgery operations.}
    \label{fig:pathwidth_opt_examples}
\end{figure}
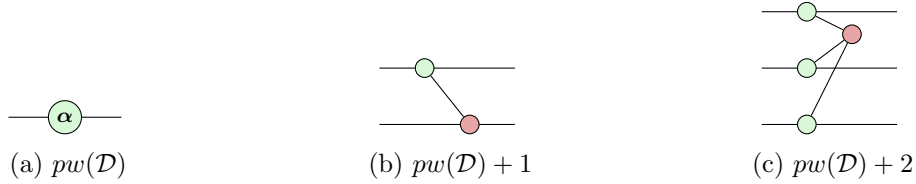

We now prove the NP-hardness of the fixed-endbags pathwidth problem by relying on its relation with the pathwidth problem.

\begin{theorem}
    The fixed-endbags pathwidth problem is NP-hard.
\end{theorem}

\begin{proof}
    Let $\mathcal{G}=(V, E)$ be a graph with vertex set $V$ and edge set $E$.
    And let $f_{U, W}$ be an optimal solution to the fixed-endbags for $\mathcal{G}$ where $U$ and $W$ are the fixed first and last bags of the ordering $f_{U, W}$.
    Then, by definition, the pathwidth of $\mathcal{G}$ satisfies
    \begin{equation}\label{eq:pathwidth_red}
        pw(\mathcal{G}) = \min_{\substack{u, w \in V\\u\neq w}} pw_{f_{\{u\}, \{w\}}}(\mathcal{G}).
    \end{equation}
    Therefore, we can solve the pathwidth problem by solving the fixed-endbags pathwidth problem a polynomial number of times, and then use these solutions to compute the pathwidth of $\mathcal{G}$ via Equation~\ref{eq:pathwidth_red} in polynomial time.
    Thus, the fixed-endbags pathwidth problem is at least as hard as the pathwidth problem, which is an NP-hard problem~\cite{kashiwabara1979np, ohtsuki1979one}.
\end{proof}

Some modifications to the initial ZX-diagram and its associated signature graph can be realized in order to solve the fixed-endbags pathwidth problem more efficiently.
Similarly to the fixed-endvertices cutwidth problem, spiders incident to two wires can be replaced by a single wire and inserted back on their associated wire in the optimized ZX-diagram without altering the obtained solution.
Also, if a vertex $u$ in a signature graph $\mathcal{G}_\mathcal{D}$ is connected to exactly one vertex $v$ associated with an input or output wire in the ZX-diagram $\mathcal{D}$, then we can remove $v$ from $\mathcal{G}_\mathcal{D}$ and replace it with $u$ in the input or output vertex set.

\section{Qubit-count optimization in quantum circuits using ZX-calculus}\label{sec:qubits_opt_quantum_circuits}

In this section, we rely on the ZX-calculus as an intermediate language to exploit the results of Section~\ref{sec:qubits_opt_lattice_surgery} and optimize the number of qubits in a quantum circuit.
The procedure for optimizing the number of qubits in a given quantum circuit $C$ using Algorithm~\ref{alg:zx_fixed_endbags_pathwidth} is as follows:
\begin{enumerate}
    \item Translate $C$ into a ZX-diagram $\mathcal{D}$ by using the translation table of Figure~\ref{fig:quantum_circuit_zx_translation}, and use the spider fusion rule to minimize the number of spiders within $\mathcal{D}$.
    \item Execute Algorithm~\ref{alg:zx_fixed_endbags_pathwidth} with $\mathcal{D}$ given as input to produce an optimized and equivalent ZX-diagram $\mathcal{D}'$.
    \item Translate $\mathcal{D}'$ into a quantum circuit $C'$ by using the translation table of Figure~\ref{fig:quantum_circuit_zx_translation}.
\end{enumerate}

\begin{figure}[t]
    \centering
    \begin{tikzpicture}
	\begin{pgfonlayer}{nodelayer}
		\node [style=none] (776) at (-9, -4.025) {};
		\node [style=none] (777) at (-9, -7.725) {};
		\node [cross, scale=0.8, style=none] (778) at (-7.925, -5.25) {};
		\node [style=Circuit label] (779) at (-8.675, -5.25) {$\ket{+}$};
		\node [style=none] (783) at (0.75, -5.25) {};
		\node [style=none] (785) at (0.75, -6.5) {};
		\node [cross, scale=0.8, style=none] (786) at (-8, -6.5) {};
		\node [style=Circuit label] (787) at (-8.75, -6.5) {$\ket{0}$};
		\node [style=Control node] (788) at (-5.5, -4) {};
		\node [style=Gate] (789) at (-5.5, -5.25) {$Z$};
		\node [meter, scale=0.9, yshift=0.16mm] (790) at (-0.5, -4.025) {};
		\node [style=none] (791) at (-0.75, -4.225) {\tiny $X$};
		\node [meter, scale=0.9, yshift=0.16mm] (792) at (-0.5, -7.775) {};
		\node [style=none] (793) at (-0.75, -7.975) {\tiny $X$};
		\node [style=Control node] (794) at (-6.75, -4) {};
		\node [style=Target node] (795) at (-6.75, -7.75) {};
		\node [style=Control node] (796) at (-4.25, -4) {};
		\node [style=Target node] (797) at (-4.25, -7.75) {};
		\node [style=Gate] (798) at (-5.5, -7.75) {$T$};
		\node [style=Control node] (799) at (-3, -7.75) {};
		\node [style=Control node] (800) at (-1.75, -5.25) {};
		\node [style=Target node] (801) at (-3, -6.5) {};
		\node [style=Target node] (802) at (-1.75, -6.5) {};
		\node [style=Gate] (803) at (-0.5, -6.5) {$T$};
		\node [style=none] (804) at (5, -4.025) {};
		\node [style=none] (805) at (5, -7.725) {};
		\node [style=none] (808) at (14.5, -5.25) {};
		\node [style=none] (809) at (14.5, -6.5) {};
		\node [style=Z dot] (812) at (8.5, -4) {};
		\node [style=Z dot] (813) at (8.5, -5.25) {};
		\node [style=Z dot] (818) at (7.25, -4) {};
		\node [style=X dot] (819) at (7.25, -7.75) {};
		\node [style=Z dot] (820) at (9.75, -4) {};
		\node [style=X dot] (821) at (9.75, -7.75) {};
		\node [style=Z phase dot] (822) at (8.5, -7.75) {$\frac{\pi}{4}$};
		\node [style=Z dot] (823) at (11, -7.75) {};
		\node [style=Z dot] (824) at (12.25, -5.25) {};
		\node [style=X dot] (825) at (11, -6.5) {};
		\node [style=X dot] (826) at (12.25, -6.5) {};
		\node [style=Z phase dot] (827) at (13.5, -6.5) {$\frac{\pi}{4}$};
		\node [style=Z dot] (828) at (13.5, -4) {};
		\node [style=Z dot] (829) at (6, -5.25) {};
		\node [style=Z dot] (830) at (13.5, -7.75) {};
		\node [style=X dot] (831) at (6, -6.5) {};
		\node [style=none] (832) at (5, -11.025) {};
		\node [style=none] (833) at (5, -12.225) {};
		\node [style=none] (834) at (14.5, -11) {};
		\node [style=none] (835) at (14.5, -12.25) {};
		\node [style=Z dot] (838) at (6, -11) {};
		\node [style=X dot] (839) at (6, -12.25) {};
		\node [style=Z dot] (840) at (8.5, -11) {};
		\node [style=X dot] (841) at (8.5, -12.25) {};
		\node [style=Z phase dot] (842) at (7.25, -12.25) {$\frac{\pi}{4}$};
		\node [style=Z dot] (844) at (12.25, -11) {};
		\node [style=X dot] (846) at (12.25, -12.25) {};
		\node [style=Z phase dot] (847) at (13.5, -12.25) {$\frac{\pi}{4}$};
		\node [style=none] (850) at (-7.75, -11.025) {};
		\node [style=none] (851) at (-7.75, -12.225) {};
		\node [style=none] (852) at (1, -11) {};
		\node [style=none] (853) at (1, -12.25) {};
		\node [style=Control node] (854) at (-6.5, -11) {};
		\node [style=Target node] (855) at (-6.5, -12.25) {};
		\node [style=Control node] (856) at (-4, -11) {};
		\node [style=Target node] (857) at (-4, -12.25) {};
		\node [style=Gate] (858) at (-5.25, -12.25) {$T$};
		\node [style=Control node] (859) at (-1.5, -11) {};
		\node [style=Target node] (860) at (-1.5, -12.25) {};
		\node [style=Gate] (861) at (-0.25, -12.25) {$T$};
		\node [style=Gate] (862) at (-2.75, -11) {$H$};
		\node [style=none] (863) at (4, -5.875) {};
		\node [style=none] (864) at (2, -5.875) {};
		\node [style=Text node] (865) at (3, -5.25) {\tiny{Corresponding\\[-3pt]ZX-diagram}};
		\node [style=none] (866) at (2, -11.625) {};
		\node [style=none] (867) at (4, -11.625) {};
		\node [style=Text node] (868) at (3, -11) {\tiny{Optimized\\[-3pt]quantum circuit}};
		\node [style=none] (869) at (9.125, -10.25) {};
		\node [style=none] (870) at (9.125, -8.75) {};
		\node [style=Text node] (871) at (10.625, -9.5) {\tiny{Optimization}};
		\node [style=Z dot] (872) at (9.75, -11) {};
		\node [style=Z dot] (873) at (11, -11) {};
	\end{pgfonlayer}
	\begin{pgfonlayer}{edgelayer}
		\draw [style=Circuit edge] (778.center) to (783.center);
		\draw [style=Circuit edge] (786.center) to (785.center);
		\draw [style=Circuit edge] (788) to (789);
		\draw [style=Circuit edge] (792) to (777.center);
		\draw [style=Circuit edge] (790) to (776.center);
		\draw [style=Circuit edge] (794) to (795);
		\draw [style=Circuit edge] (796) to (797);
		\draw [style=Circuit edge] (801) to (799);
		\draw [style=Circuit edge] (802) to (800);
		\draw [style=Hadamard edge] (812) to (813);
		\draw (818) to (819);
		\draw (820) to (821);
		\draw (825) to (823);
		\draw (826) to (824);
		\draw (829) to (808.center);
		\draw (828) to (804.center);
		\draw (830) to (805.center);
		\draw (831) to (809.center);
		\draw (838) to (839);
		\draw (840) to (841);
		\draw (846) to (844);
		\draw (833.center) to (835.center);
		\draw [style=Circuit edge] (852.center) to (850.center);
		\draw [style=Circuit edge] (859) to (860);
		\draw [style=Circuit edge] (856) to (857);
		\draw [style=Circuit edge] (854) to (855);
		\draw [style=Circuit edge] (851.center) to (853.center);
		\draw (832.center) to (840);
		\draw (844) to (834.center);
		\draw [style=Arrow] (864.center) to (863.center);
		\draw [style=Arrow] (867.center) to (866.center);
		\draw [style=Arrow] (870.center) to (869.center);
		\draw [style=Hadamard edge] (873) to (872);
		\draw (872) to (840);
		\draw (873) to (844);
	\end{pgfonlayer}
\end{tikzpicture}
    \caption{Overview of the process described for optimizing the number of qubits in a quantum circuit.
    The initial quantum circuit is translated into a ZX-diagram.
    Then, the ZX-diagram is optimized using the method explained in Section~\ref{sec:qubits_opt_lattice_surgery}.
    Finally, the optimized ZX-diagram is translated back into a quantum circuit with an optimized number of qubits.}
    \label{fig:circuit_qubits_opt_example}
\end{figure}
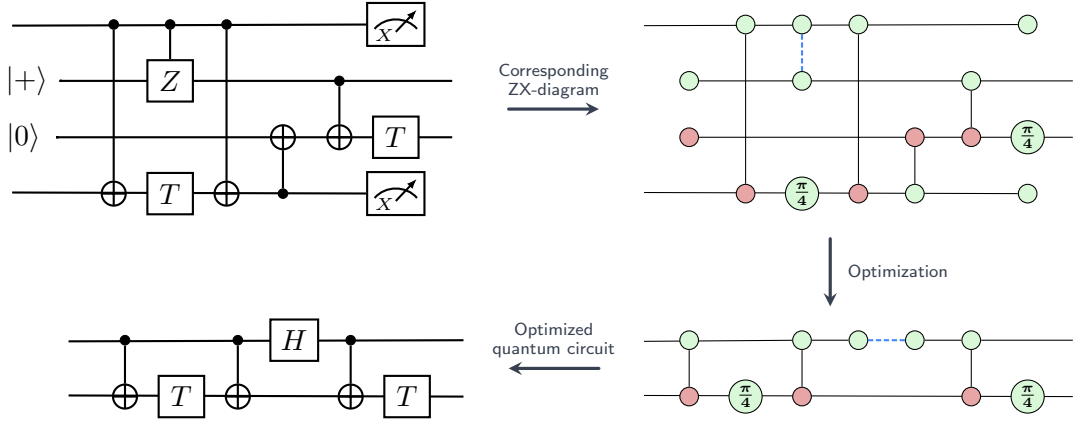

An example is provided in Figure~\ref{fig:circuit_qubits_opt_example}.
Translating a quantum circuit into a ZX-diagram can always easily be done, for example by using the procedure described in Section~\ref{sub:zx_calculus}.
The spider fusion rule is applied to minimize the number of spiders within $\mathcal{D}$ in order to improve the effectiveness of Algorithm~\ref{alg:zx_fixed_endbags_pathwidth}.
Note that a more sophisticated procedure, using other rules than the spider fusion rule, could be used here to transform $\mathcal{D}$ before executing Algorithm~\ref{alg:zx_fixed_endbags_pathwidth}.
We leave it as a question for future research to determine which transformation procedure for $\mathcal{D}$ could be use to maximize the effectiveness of Algorithm~\ref{alg:zx_fixed_endbags_pathwidth}.
The last step of the procedure consists in translating the optimized ZX-diagram $\mathcal{D}'$ into a quantum circuit $C'$.
We can notice that a non-horizontal wire in $\mathcal{D}'$ is either connecting a Z spider to an X spider, or is a dashed blue edge connected two Z spiders.
In the first case, it can be translated to a CNOT gate; and in the second case, it can be translated to a $CZ$ gate.
And all the horizontal wires in $\mathcal{D}$ can be mapped to horizontal wires in $C'$, where they represent qubits.
Thus, we obtain a quantum circuit $C'$ which has a number of qubits characterized by the following theorem.

\begin{theorem}
    Let $\mathcal{D}$ be a ZX-diagram in which no spider of the same color are connected by a wire.
    Let $\mathcal{D}'$ be a ZX-diagram produced by Algorithm~\ref{alg:zx_fixed_endbags_pathwidth} when $\mathcal{D}$ is given as input and let $C'$ be the quantum circuit translated from $\mathcal{D}'$ using the translation table of Figure~\ref{fig:quantum_circuit_zx_translation}.
    Then the number of qubits in $C'$ is at least equal to $pw(\mathcal{D})$ and at most equal to $pw(\mathcal{D}) + 1$.
\end{theorem}

\begin{proof}
    Each non-horizontal wire in $\mathcal{D}'$ is either connected a Z spider to a X spider or is a dashed blue edge connected two Z spiders.
    Therefore, each non-horizontal wire is translated as either a CNOT gate or a $CZ$ gate in $C'$.
    As stated by Theorem~\ref{thm:pathwidth_lower_bound}, the maximum number of horizontal wires in any vertical cut of $\mathcal{D}'$ is at least equal to $pw(\mathcal{D})$.
    And as demonstrated in the proof of Theorem~\ref{thm:pathwidth_upper_bound}, the maximum number of horizontal wires in any vertical cut of $\mathcal{D}'$ is at most equal to $pw(\mathcal{D}) + 1$.
    Each horizontal wire in $\mathcal{D}'$ is translated as a qubit in $C'$.
    Therefore, the number of qubits in $C'$ is at least equal to $pw(\mathcal{D})$ and at most equal to $pw(\mathcal{D}) + 1$.
\end{proof}

\section{Benchmarks}\label{sec:bench}

In this section, we evaluate the qubit-count reductions achieved by our approaches on a set of quantum circuits in which the number of $T$ gates has been optimized.
To solve the Hadamard gates degadgetization problem, we used the optimal solver of the directed feedback vertex set problem described in Reference~\cite{bathie2022pace}.
For the fixed-endbags pathwidth problem, we applied a fast but non-optimal greedy algorithm.
Our implementations of the algorithms used for the benchmarks are open source~\cite{github}.

The quantum circuits were obtained from References~\cite{amyGithub} and~\cite{reversibleBenchmarks}.
These circuits were pre-optimized as follows.
First, the \texttt{FastTMerge} algorithm of Reference~\cite{vandaele2024optimalnumberparametrizedrotations} was applied to rapidly and efficiently reduce the number of $T$ gates in the circuits.
This method optimally minimizes the number of non-Clifford gates in Clifford+$R_Z$ circuits when no information about the angles of the $R_Z$ gates is known~\cite{vandaele2024optimalnumberparametrizedrotations}.
Then, the \texttt{InternalHOpt} algorithm of Reference~\cite{vandaele2024optimal} was applied to optimize the number of internal Hadamard gates in the circuits.
This method is also optimal for Clifford+$R_Z$ circuits when no information about the angles of the $R_Z$ gates is known~\cite{vandaele2024optimalnumberparametrizedrotations}.
Finally, the Hadamard gates of the circuits were gadgetized and the \texttt{FastTODD} algorithm of Reference~\cite{vandaele2024lower} was applied to optimize the number of $T$ gates in the circuits.
We chose the \texttt{FastTODD} $T$-count optimizer of Reference~\cite{vandaele2024lower} as it currently yields the best results on most of the quantum circuits, while being faster than other $T$-count optimizers.
Note that the benchmark results may vary depending on the pre-optimization procedure used, such as when a different $T$-count optimizer is applied.

The benchmark results are presented in Table~\ref{tab:bench_qubits_opt}.
We can notice that the approach based on the fixed-endbags pathwidth problem performs as well as or surpasses the Hadamard gates degadgetization approach on most quantum circuits evaluated.
Nevertheless, the Hadamard gates degadgetization approach is performing better than the approach based on the fixed-endbags pathwidth problem on a few quantum circuits.
This discrepancy is probably due to the non-optimality of the greedy algorithm used in the benchmarks for the fixed-endbags pathwidth problem.
As future work, it would be valuable to develop more efficient algorithms for tackling the fixed-endbags pathwidth problem.
An optimal algorithm for the fixed-endbags pathwidth problem would certainly have a large complexity due to the NP-hardness of the problem but could potentially be executed on small circuits.

\begin{table}
    \centering
\begin{tabular}{lccccccc}
        \toprule
         & \multicolumn{3}{c}{Pre-optimization} && \multicolumn{3}{c}{Qubit-count} \\
        \cmidrule(lr){2-4} \cmidrule(lr){6-8}
        Circuit & $n$ & $h$ & Qubit-count && Degadgetization && Pathwidth \\
        \midrule
        Adder$_8$ & 24 & 37 & 61 && 58 && 55 \\
        Barenco Tof$_3$ & 5 & 3 & 8 && 7 && 7 \\
        Barenco Tof$_4$ & 7 & 7 & 14 && 11 && 10 \\
        Barenco Tof$_5$ & 9 & 11 & 20 && 13 && 12 \\
        Barenco Tof$_{10}$ & 19 & 31 & 50 && 36 && 31 \\
        CSLA MUX$_3$ & 15 & 6 & 21 && 20 && 20 \\
        CSUM MUX$_9$ & 30 & 12 & 42 && 40 && 38 \\
        Grover$_5$ & 9 & 68 & 77 && 57 && 33 \\
        Ham$_{15}$ (high) & 20 & 331 & 351 && 175 && 123 \\
        Ham$_{15}$ (low) & 17 & 29 & 46 && 35 && 31 \\
        Ham$_{15}$ (med) & 17 & 54 & 71 && 46 && 37 \\
        Mod Adder$_{1024}$ & 28 & 304 & 332 && 246 && 198 \\
        Mod Mult$_{55}$ & 9 & 3 & 12 && 10 && 12 \\
        Mod Red$_{21}$ & 11 & 17 & 28 && 24 && 19 \\
        QCLA Adder$_{10}$ & 36 & 25 & 61 && 55 && 49 \\
        QCLA Com$_7$ & 24 & 18 & 42 && 36 && 32 \\
        QCLA Mod$_7$ & 26 & 58 & 84 && 71 && 41 \\
        QFT$_{4}$ & 5 & 38 & 43 && 8 && 8 \\
        RC Adder$_6$ & 14 & 10 & 24 && 22 && 23 \\
        Tof$_3$ & 5 & 2 & 7 && 6 && 6 \\
        Tof$_4$ & 7 & 4 & 11 && 9 && 10 \\
        Tof$_5$ & 9 & 6 & 15 && 13 && 12 \\
        Tof$_{10}$ & 19 & 16 & 35 && 27 && 25 \\
        VBE Adder$_3$ & 10 & 4 & 14 && 13 && 13 \\
        \bottomrule
\end{tabular}
\caption{Qubit-count achieved by different methods on a set of quantum circuits with an optimized number of $T$ gates.
    The values $n$ and $h$ are respectively corresponding to the number of qubits and the number of internal Hadamard gates gadgetized in the initial circuit.
    The qubit-count in the initial circuit is equal to $n+h$.
    The degadgetization column refers to the number of qubits achieved by the Hadamard gates degadgetization method.
    The pathwidth column refers to the number of qubits achieved by the approach based on the fixed-endbags pathwidth problem.
    All circuits were processed by the algorithms in no more than a few seconds.}
    \label{tab:bench_qubits_opt}
\end{table}

\section{Perspectives}\label{sec:perspectives}

This section outlines several potential avenues for future research work based on the results presented herein.

\subsection{Correction strategy and depth optimization}

To deterministically realize an operator in lattice surgery, it is required to have a correction strategy in order to counteract the non-deterministic nature of some lattice surgery operations.
This correction strategy imposes a specific order in which the lattice surgery operations must be realized.
Therefore, while necessary for determinism, the correction strategy is limiting the actions that can be performed to optimize the number of qubits.
A method for finding a correction strategy, called PF-flow, for an arbitrary ZX-diagram has been proposed in Reference~\cite{de2019pauli}.
However, the completeness of this approach has not been proved.
This means that there may be some ZX-diagrams which don't have a PF-flow but which can still be implemented deterministically trough lattice surgery operations by using another correction strategy.
Therefore, an open problem is to find a complete correction strategy, and to better characterize the set of ZX-diagrams which are deterministically implementable using lattice surgery operations.
Addressing these challenges could potentially lead to better optimization in the qubit-count.

Furthermore, the PF-flow has been designed to optimize the computational depth, as it aims at finding a correction strategy that divides the ZX-diagram in a minimized number of layers.
A simple method for optimizing both the depth and the qubit-count could consists in finding a PF-flow and then applying our methods for optimizing the number of qubits in each layer of the ZX-diagram.
However, this approach does not minimize the number of qubits in between two layers of the ZX-diagram.
An open problem is to develop a more comprehensive strategy to simultaneously optimize both the computational depth and the number of qubits.

\subsection{Qubit-count optimization using a complete ZX-calculus equational theory}

We proposed an approach to optimize the number of logical qubits required to implement a given ZX-diagram through lattice surgery operations.
Our optimization procedure rearranges the order of the spiders and makes use of the identity and fusion rules of the ZX-calculus.
A natural extension would be to consider other ZX-calculus rules that could lead to even better optimization.
For instance, applying our algorithm to the following ZX-diagram, which represents a phase gadget, yields a ZX-diagram for which $5$ logical qubits are required to realize the lattice surgery operations it represents:

\begin{equation}
    \begin{tikzpicture}
	\begin{pgfonlayer}{nodelayer}
		\node [style=none] (0) at (-1, 3) {};
		\node [style=none] (1) at (-1, 1.75) {};
		\node [style=none] (2) at (-1, 0.5) {};
		\node [style=none] (3) at (3, 3) {};
		\node [style=none] (4) at (3, 1.75) {};
		\node [style=none] (5) at (3, 0.5) {};
		\node [style=Z dot] (6) at (0, 3) {};
		\node [style=Z dot] (7) at (0, 1.75) {};
		\node [style=Z dot] (8) at (0, 0.5) {};
		\node [style=X dot] (9) at (1, 2.375) {};
		\node [style=Z phase dot] (10) at (2, 2.375) {$\alpha$};
		\node [style=none] (11) at (4.5, 1.75) {$=$};
		\node [style=none] (12) at (5.5, 3) {};
		\node [style=none] (13) at (5.5, 1.75) {};
		\node [style=none] (14) at (5.5, 0.5) {};
		\node [style=none] (15) at (10.25, 3) {};
		\node [style=none] (16) at (10.25, 1.75) {};
		\node [style=none] (17) at (10.25, 0.5) {};
		\node [style=Z dot] (18) at (6.5, 3) {};
		\node [style=Z dot] (19) at (6.5, 1.75) {};
		\node [style=Z dot] (20) at (7.25, 0.5) {};
		\node [style=X dot] (21) at (7.25, 2.375) {};
		\node [style=Z phase dot] (22) at (9.25, 2.375) {$\alpha$};
		\node [style=X dot] (23) at (8.25, 2.375) {};
	\end{pgfonlayer}
	\begin{pgfonlayer}{edgelayer}
		\draw (0.center) to (3.center);
		\draw (1.center) to (4.center);
		\draw (2.center) to (5.center);
		\draw (6) to (9);
		\draw (9) to (7);
		\draw (8) to (9);
		\draw (9) to (10);
		\draw (12.center) to (15.center);
		\draw (13.center) to (16.center);
		\draw (14.center) to (17.center);
		\draw (18) to (21);
		\draw (21) to (19);
		\draw (21) to (22);
		\draw (23) to (20);
	\end{pgfonlayer}
\end{tikzpicture}
\end{equation}
However, as shown below, there exists an equivalent ZX-diagram which can be implemented in lattice surgery with only $4$ logical qubits:

\begin{equation}
    \begin{tikzpicture}
	\begin{pgfonlayer}{nodelayer}
		\node [style=none] (0) at (-1, 3) {};
		\node [style=none] (1) at (-1, 1.75) {};
		\node [style=none] (2) at (-1, 0.5) {};
		\node [style=none] (3) at (3, 3) {};
		\node [style=none] (4) at (3, 1.75) {};
		\node [style=none] (5) at (3, 0.5) {};
		\node [style=Z dot] (6) at (0, 3) {};
		\node [style=Z dot] (7) at (0, 1.75) {};
		\node [style=Z dot] (8) at (0, 0.5) {};
		\node [style=X dot] (9) at (1, 2.375) {};
		\node [style=Z phase dot] (10) at (2, 2.375) {$\alpha$};
		\node [style=none] (11) at (4.5, 1.75) {$=$};
		\node [style=none] (12) at (5.5, 3) {};
		\node [style=none] (13) at (5.5, 1.75) {};
		\node [style=none] (14) at (5.5, 0.5) {};
		\node [style=none] (15) at (15.5, 3) {};
		\node [style=none] (16) at (15.5, 1.75) {};
		\node [style=none] (17) at (15.5, 0.5) {};
		\node [style=Z dot] (18) at (6.5, 3) {};
		\node [style=X dot] (19) at (7.5, 1.75) {};
		\node [style=Z phase dot] (22) at (10.5, 0.5) {$\alpha$};
		\node [style=Z dot] (23) at (8.5, 1.75) {};
		\node [style=X dot] (24) at (9.5, 0.5) {};
		\node [style=Z dot] (25) at (14.5, 3) {};
		\node [style=X dot] (26) at (13.5, 1.75) {};
		\node [style=Z dot] (27) at (12.5, 1.75) {};
		\node [style=X dot] (28) at (11.5, 0.5) {};
	\end{pgfonlayer}
	\begin{pgfonlayer}{edgelayer}
		\draw (0.center) to (3.center);
		\draw (1.center) to (4.center);
		\draw (2.center) to (5.center);
		\draw (6) to (9);
		\draw (9) to (7);
		\draw (8) to (9);
		\draw (9) to (10);
		\draw (12.center) to (15.center);
		\draw (13.center) to (16.center);
		\draw (14.center) to (17.center);
		\draw (18) to (19);
		\draw (23) to (24);
		\draw (25) to (26);
		\draw (27) to (28);
	\end{pgfonlayer}
\end{tikzpicture}
\end{equation}

Ultimately, our goal is to rely on a complete ZX-calculus equational theory to find an equivalent ZX-diagram which can be implemented in lattice surgery by using a minimal number of qubits.
Our results indicates that this consists in finding an equivalent ZX-diagram $\mathcal{D}$ such that $pw(\mathcal{D})$ is minimized.
Additionally, we would like to maintain the ability to translate the optimized ZX-diagram directly into a quantum circuit, akin to the straightforward translation enabled by the proposed approach based on the fixed-endbags pathwidth problem.
In our approach, this is guaranteed by the fact that each non-horizontal wire corresponds either to a dashed blue edge connecting two Z spiders, representing a $CZ$ gate, or a wire connected a Z spider and a X spider, representing a CNOT gate.
For other cases, an intermediate qubit may be required to translate the ZX-diagram into a quantum circuit, for example:

\begin{equation}
    \begin{tikzpicture}
	\begin{pgfonlayer}{nodelayer}
		\node [style=none] (0) at (-4, 2) {};
		\node [style=none] (1) at (-2, 2) {};
		\node [style=none] (2) at (-4, 0.75) {};
		\node [style=none] (3) at (-2, 0.75) {};
		\node [style=Z dot] (4) at (-3, 2) {};
		\node [style=Z dot] (5) at (-3, 0.75) {};
		\node [style=none] (6) at (0, 2.625) {};
		\node [style=none] (7) at (3.5, 2.625) {};
		\node [style=none] (8) at (0, 0.125) {};
		\node [style=none] (9) at (3.5, 0.125) {};
		\node [style=Z dot] (10) at (1.25, 2.625) {};
		\node [style=X dot] (12) at (0.25, 1.375) {};
		\node [style=X dot] (13) at (1.25, 1.375) {};
		\node [style=X dot] (14) at (2.25, 1.375) {};
		\node [style=Z dot] (15) at (2.25, 0.125) {};
		\node [style=X dot] (16) at (3.25, 1.375) {};
		\node [style=none] (17) at (7, 2.625) {};
		\node [style=none] (18) at (11.25, 2.625) {};
		\node [style=none] (19) at (7, 0.125) {};
		\node [style=none] (20) at (11.25, 0.125) {};
		\node [style=Control node] (21) at (8, 2.625) {};
		\node [style=none] (22) at (7, 1.375) {};
		\node [style=Target node] (23) at (8, 1.375) {};
		\node [style=Target node] (24) at (9.25, 1.375) {};
		\node [style=Control node] (25) at (9.25, 0.125) {};
		\node [style=none] (27) at (4.75, 1.375) {$\Longleftrightarrow$};
		\node [style=none] (28) at (-1, 1.375) {$=$};
		\node [meter, scale=0.9, yshift=0.16mm] (29) at (10.75, 1.375) {};
		\node [style=none] (30) at (10.5, 1.175) {\tiny $Z$};
		\node [style=Circuit label] (31) at (6.25, 1.375) {$\ket{0}$};
	\end{pgfonlayer}
	\begin{pgfonlayer}{edgelayer}
		\draw (3.center) to (2.center);
		\draw (0.center) to (1.center);
		\draw (4) to (5);
		\draw (9.center) to (8.center);
		\draw (6.center) to (7.center);
		\draw (15) to (14);
		\draw (10) to (13);
		\draw (12) to (14);
		\draw (16) to (14);
		\draw [style=Circuit edge] (20.center) to (19.center);
		\draw [style=Circuit edge] (17.center) to (18.center);
		\draw [style=Circuit edge] (25) to (24);
		\draw [style=Circuit edge] (21) to (23);
		\draw [style=Circuit edge] (22.center) to (24);
		\draw [style=Circuit edge] (24) to (29);
	\end{pgfonlayer}
\end{tikzpicture}
\end{equation}

\section{Conclusion}
We proposed novel methods for optimizing the number of qubits in quantum circuits and the number of logical qubits required to perform lattice surgery operations.
When evaluated on a set of quantum circuits in which the number of $T$ gates has been optimized, our methods lead to a reduction of up to 65\% in the number of qubits.
On the basis of our results, we suggested several promising avenues for future research work.
These include extending the proposed approach to potentially achieve even better reduction in the number of qubits, and investigating how our approach could be adapted to take into account other important metrics such as the computational depth.

Our work make use of the ZX-calculus and its fundamental rules to seamlessly formulate the qubit-count optimization problem and establish a connection with well-studied problems in graph theory.
As such, our findings provide a compelling case for the use of the ZX-calculus in quantum compilation.
This further motivates the adoption of the ZX-calculus as a formal, powerful, and versatile tool for representing and optimizing quantum computation.

\bibliographystyle{unsrt}
\bibliography{ref.bib}

\end{document}